\documentclass[11pt]{article} 
\usepackage{amsmath,amsfonts,bm}
\usepackage{amsthm}
\usepackage{mathtools}

\newtheorem{theorem}{Theorem}[section]

\newtheorem{fact}[theorem]{Fact}
\newtheorem{lemma}[theorem]{Lemma}

\newtheorem{definition}{Definition}[section]

        {\hspace*{\fill}$\Box$\par}

\DeclarePairedDelimiterX{\infdivx}[2]{(}{)}{%
  #1\;\delimsize\|\;#2%
}

\def\eqref#1{equation~\ref{#1}}
\def\1{\bm{1}}

\def\eps{{\varepsilon}}

\def\rvv{{\mathbf{v}}}

\def\vv{{\bm{v}}}
\def\vw{{\bm{w}}}

\def\mA{{\bm{A}}}
\def\mB{{\bm{B}}}

\def\mG{{\bm{G}}}
\def\mH{{\bm{H}}}

\def\mK{{\bm{K}}}

\def\mQ{{\bm{Q}}}

\def\mV{{\bm{V}}}
\def\mW{{\bm{W}}}
\def\mX{{\bm{X}}}

\DeclareMathAlphabet{\mathsfit}{\encodingdefault}{\sfdefault}{m}{sl}
\SetMathAlphabet{\mathsfit}{bold}{\encodingdefault}{\sfdefault}{bx}{n}

\def\gN{{\mathcal{N}}}

\def\gW{{\mathcal{W}}}

\newcommand{\E}{\mathbb{E}}

\newcommand{\R}{\mathbb{R}}

\newcommand{\tr}{\mathrm{tr}}
\DeclareMathOperator{\Tr}{Tr}

\renewcommand{\i}{\mathbf{i}}

\newcommand{\tail}{\operatorname{tail}}
 \renewcommand{\epsilon}{\varepsilon}

\usepackage{paper}

\title{Optimal Query Complexities for Dynamic Trace Estimation}

\author{David P. Woodruff\footnote{Work done while at Google Research in Pittsburgh.} \\ Carnegie Mellon University \\ \texttt{dwoodruf@cs.cmu.edu} \and Fred Zhang\footnote{Work done while the author was a resaerch intern at Google Brain.} \\ UC Berkeley \\ \texttt{z0@berkeley.edu} \\ \and Qiuyi (Richard) Zhang  \\ Google Brain \\ \texttt{qiuyiz@google.com}}

\date{}

\begin{document}

\maketitle

\begin{abstract}
We consider the problem of minimizing the number of matrix-vector queries needed for accurate trace estimation in the dynamic setting where our underlying matrix is changing slowly, such as during an optimization process. Specifically, for any $m$ matrices $\mA_1,...,\mA_m$ with consecutive differences bounded in Schatten-$1$ norm by $\alpha$, we provide a novel binary tree summation procedure that simultaneously estimates all $m$ traces up to $\epsilon$ error with $\delta$ failure probability with an optimal query complexity of $\widetilde{O}(m \alpha\sqrt{\log(1/\delta)}/\epsilon + m\log(1/\delta))$, improving the dependence on both $\alpha$ and $\delta$ from Dharangutte and Musco (NeurIPS, 2021). Our procedure works without additional norm bounds on $\mA_i$ and can be generalized to a bound for the $p$-th Schatten norm  %$\alpha_p$ 
for $p \in [1,2]$, giving a complexity of $\widetilde{O}(m \alpha(\sqrt{\log(1/\delta)}/\epsilon)^p +m \log(1/\delta))$. 

By using novel reductions to communication complexity and information-theoretic analyses of Gaussian matrices, we provide matching lower bounds for static and dynamic trace estimation in all relevant parameters, including the failure probability. Our lower bounds (1) give the first tight bounds for Hutchinson's estimator in the matrix-vector product model with Frobenius norm error {\it even in the static setting}, and (2) are the first unconditional lower bounds for dynamic trace estimation, resolving open questions of prior work.
\end{abstract}

\newpage
\section{Introduction}
\label{sec:intro}
Implicit matrix trace estimation is ubiquitous in numerical linear algebra and arises naturally in a wide range of applications, see, e.g.,~\cite{ubaru2018trace_est_applications}. 
In this problem, we are given an oracle which gives us matrix-vector products $\mA x_1,\mA x_2,\cdots ,\mA x_m$ for an unknown $n\times n$ square matrix $\mA$ and queries $x_1, \ldots, x_m$ of our choice, that may be chosen adaptively. In typical applications, one cannot afford to compute the diagonal entries of $\mA$ explicitly, due to $\mA$ being implicitly represented and computational constraints. The goal is to efficiently estimate   $\Tr \mA$ using only matrix-vector products.

In machine learning and data science, applications of trace estimation include training Gaussian Processes~\cite{dong2017log_determinant_gp, fgcofr2017entropy_logdet},   triangle counting \cite{avron2010triangle_counting},    computing the Estrada Index~\cite{eh2008graph_estrada_index, estrada2000characterization}, and studying   optimization landscapes  of deep neural networks from Hessian matrices~\cite{gkx2019investigation, ygkm2020pyhessian}. 
In these applications, it is common that $\mA$ is represented implicitly due to its large memory footprint. For example, while it is possible to compute Hessian-vector products via Pearlmutter's trick~\cite{pearlmutter1994hv_trick}, it is prohibitive to compute or store the   Hessian matrix $\mH$, see, e.g., ~\cite{gkx2019investigation}. 

Moreover, $\mA$ may be a matrix function $f$ of another matrix $\mB$ in some applications. Since computing $f(\mB)$ is expensive, it is desirable to apply implicit trace estimation. 
For example, during the training of Gaussian Processes, 
the marginal log-likelihood contains a heavy-computation
term, i.e., the log of the determinant of the covariance matrix, $\log(\det(\mK))$, where $\mK \in \mathbb{R}^{n\times n}$ and $n$ is the number of data points. The canonical way of computing $\log(\det(\mK))$ is via a Cholesky factorization on $\mK$, which  takes   $O(n^3)$ time. Instead, implicit trace estimation methods   provide fast algorithms for approximating $\log(\det(\mK)) = \sum_{i=1}^{n}\log(\lambda_i) = \tr(\log(\mK))$ on large-scale data. Therefore, it is important to understand the fundamental limits of implicit trace estimation as the {\it query complexity}, i.e., the minimum number of matrix-vector multiplications required to achieve a desired accuracy and success rate.

\paragraph{Static trace estimation and Hutchinson's method}
On the algorithmic side, Hutchinson's method~\cite{hut1990hutchisons_method} is a simple and  widely used method for trace estimation. Let $\mQ = [q_1, \dots, q_\ell] \in \mathbb{R}^{n \times \ell}$ be $\ell$ vectors with i.i.d.\ standard Gaussian or Rademacher random variables. Given matrix-vector multiplication access to $\mA$, Hutchinson's method estimates $\tr(\mA)$ by $t= \frac{1}{q}\sum_{i=1}^{q} q_i^T \mA q_i = \frac{1}{q} \tr(\mQ^T \mA \mQ)$.  It is known \cite{at2011trace_est_psd} that the   estimator satisfies that for any $\eps,\delta \in (0,1)$,
\begin{equation}\label{eqn:static}
    |t - \Tr \mA| \leq \eps \|\mA\|_F, \text{ with probability at least }  1-\delta,
\end{equation}
 provided the number $\ell$ of queries satisfies $\ell \geq C \log(1/\delta)/\eps^2$ for some fixed constant $C$. 
 
For Hutchinson's method, there is also previous work which showed for queries of the form $x^\top \mA x$, $\Omega(1/\epsilon^2)$ queries are required \cite{roosta2015improved}; however, this does not imply even a lower bound for non-adaptive algorithms that use matrix-vector queries. Though stronger algorithmic results and matching lower bounds are known for the important case of PSD matrices in the non-adaptive setting \cite{mmmw2020hutch_pp, jiang2021optimal}, the optimality of Hutchinson's estimator as an trace estimator for general square matrices in the matrix-vector product model still remains an open problem. Notably, Hutchinson's method chooses the query vectors non-adaptively and it is furthermore unclear whether adaptivity could help. 

More generally, there has been a flurry of recent work that gives trace estimators with $o(1/\epsilon^2)$ query complexity but with a different error guarantee. Specifically, let us consider a Schatten-$p$ norm error guarantee, where the goal is to provide an estimate $t$ such that 
\begin{equation}\label{eqn:static-p}
    |t - \Tr \mA| \leq \eps \|\mA\|_p, \text{ with probability at least }  1-\delta,
\end{equation}
where $\|\mA\|_p$ denotes the Schatten-$p$ norm. 

For $p=1$, a previous work~\cite{mmmw2020hutch_pp} 
proposes a variance-reduced version of Hutchinson's method that uses only $O(1/\epsilon)$ matrix-vector product queries  to achieve a nuclear norm error of $\eps\|\mA\|_*$, in contrast to the $O(1/\epsilon^2)$   queries used when the error is in the Frobenius norm. When the matrix is positive semidefinite (PSD), the nuclear norm error is equivalent to a $(1+\eps)$ multiplicative approximation to the trace. Their work, along with a subsequent work \cite{jiang2021optimal}, shows that $\Omega(1/\eps)$ queries are therefore sufficient and necessary to achieve a $(1+\eps)$ multiplicative trace approximation in this setting. While this line of work mainly focuses on PSD matrices and nuclear norm error, we consider trace estimation on general square matrices with  Schatten-$p$ norm error for any $p \in [1,2]$.

Furthermore, we note that the variance-reduced Hutchinson's method splits the queries between approximating the top $O(1/\epsilon)$ eigenvalues, i.e., by computing a rank-$O(1/\epsilon)$ approximation to $\mA$, and performing Hutchinsons's method on the remainder. Due to the low rank approximation subroutine, the query complexity's dependence on the failure probability is more concretely $O(\sqrt{\log(1/\delta)}/\epsilon + \log(1/\delta))$ for additive $\eps \|\mA\|_*$ error. The additive $\log(1/\delta)$ rate is shown to be necessary when non-adaptive queries are used, but it is an open problem whether adaptive queries can remove the additive $\log(1/\delta)$ term for trace estimation with Schatten-$p$ norm error \cite{jiang2021optimal}.

This motivates the natural question:

\begin{quote}
   \emph{Question 1: Is Hutchinson's method optimal in terms of $\epsilon$ and $\delta$ for static trace estimation of general square matrices, even when adaptivity is allowed? How do we generalize Hutchinson's method for error in general Schatten-$p$ norms?}
\end{quote}

%The key observation is that the variance of the estimated trace in Hutchinson's method is largest when there is a large gap between the top few eigenvalues and the remaining ones. Thus, by splitting the number of matrix-vector queries between approximating the top $O(1/\epsilon)$ eigenvalues, i.e., by computing a rank-$O(1/\epsilon)$ approximation to $\mA$, and performing trace estimation on the remaining part of the spectrum, one needs only $O(1/\epsilon)$ queries in total to achieve $\epsilon \|\mA\|_\star$ error.

% \textcolor{red}{Ricahrd: Should we expand on Question 1?} \fred{I agree. I think we should emphasize that the prior work on PSD matrices gives lower bounds for nuclear norm error.}

\paragraph{Dynamic trace estimation}  
In various applications the input matrix is not fixed. For example, during model training, we need to estimate the trace of a dynamically changing Hessian matrix  with respect to some loss function. One may assume that the change at each step is not very large. Motivated by such a scenario, a recent work by Dharangutte and Musco \cite{dharangutte2021dynamic} studies dynamic trace estimation. 

Formally, let $p\in [1,2]$ and $\mA_1,\mA_2,\cdots, \mA_m$ be $n\times n$ matrices in a stream such that (1) $\|\mA_i\|_p \leq 1 $ for all $i > 1$, where $\|\cdot\|_p$ denotes the Schatten-$p$ norm, and (2) $\|\mA_{i+1} -\mA_{i}\|_p \leq \alpha < 1$ for all $i \leq m-1$.  The goal is to output a sequence of estimates $t_1,\cdots, t_m$ such that for each $i \in [m]$,
\begin{equation}\label{eqn:dynamic}
    |t_i - \Tr \mA_i| \leq \eps , \text{ with probability at least }  1-\delta,
\end{equation}
via matrix-vector multiplication query access to the first $i$ matrices $(\mA_j)_{j=1}^i$. 
Na\"ively, one could estimate each $\Tr \mA_i$ independently using Hutchinson's method. This, however, does not exploit that the changes are bounded at each step. Alternatively, one can rewrite $\Tr \mA_i$ as $\Tr \mA_1 + \sum_{i=2}^{i} \Tr (\bm{\Delta}_i)$, where $\bm{\Delta}_i = \mA_i - \mA_{i-1}$, by linearity of the trace, and apply Hutchinson's method on each term. Unfortunately, this scheme suffers from an accumulation of errors  over the steps.

The prior work \cite{dharangutte2021dynamic}  is focused on $p=\{1,2\}$ and improves upon the na\"ive ideas above.
For $p=1$, the authors give a method that uses $O\left(m\sqrt{\alpha/\delta}/\eps + \sqrt{1/\delta}/\eps\right)$ queries. For $p=2$, they provide an algorithm with query complexity $O(m\alpha\log(1/\delta)/\eps^2 + \log(1/\delta)/\eps^2)$ and a \textit{conditional} lower bound showing that this is tight. This leaves open the question: 
\begin{quote}
   \emph{Question 2: Can we design improved algorithms for dynamic trace estimation under a general Schatten norm assumption? Can we prove an  unconditionally optimal lower bound?}
\end{quote}
%\textbf{Remark:} We remark that our techniques can relax the first condition to assume that $\|\mA_1\|_p \leq 1$, which is generally a more realistic assumption in practice. 
 
\subsection{Our Results}

\begin{table}[t]
    \centering
    \scalebox{0.8}{
    \begin{tabular}{|c|c|c|c|c|}
    \hline
    \multicolumn{5}{|c|}{\textbf{Upper Bounds}}\\
    \hline
        Prior Work & Query Complexity & Matrix Type & Failure Rate & Algorithm Type \\
    \hline
        \cite{at2011trace_est_psd, ra2015implicit_trace_est} & $O(\log(1/\delta)/\epsilon^2)$ & general square & $\delta$ & non-adaptive, $p = 2$\\
    \hline
        \cite{mmmw2020hutch_pp}  & $O(\sqrt{\log(1/\delta)}/\epsilon + \log(1/\delta))$ &  PSD & $\delta$ & adaptive, $p=1$\\
    \hline
        \cite{mmmw2020hutch_pp} & $O(\log(1/\delta)/\epsilon)$ & PSD & $\delta$ & non-adaptive, $p = 1$\\
    \hline
        \cite{jiang2021optimal} & $O(\sqrt{\log(1/\delta)}/\epsilon + \log(1/\delta))$ & PSD & $\delta$ & non-adaptive, $p = 1$\\
    \hline
    \textbf{This work \textsuperscript{1}} & $O((\sqrt{\log(1/\delta)}/\epsilon)^p + \log(1/\delta))$ & PSD & $\delta$ & non-adaptive, general $p$ \\
    \hline
    \hline
        \multicolumn{5}{|c|}{\textbf{Lower Bounds (Adaptive)}}\\
    \hline
        \cite{mmmw2020hutch_pp} & $\Omega(1/(b+\eps\log(1/\eps)))$ & general square, bit & constant & adaptive, $p = 1$\\
    
    \hline
    \textbf{This work   \textsuperscript{2}}  & $\Omega\left(\frac{1}{\epsilon^p (b + \log(1/\epsilon)) }   + \frac{\log (1/\delta)}{(b + \log\log(1/\delta))} \right)$  & general square, bit & $\delta$ & adaptive, general $p$ \\
    \hline
     \textbf{This work   \textsuperscript{3}} & $\Omega\left((\sqrt{\log(1/\delta)}/\eps)^p \right)$  & general square, ram & $\delta$ & adaptive, general $p$\\
    \hline
    \multicolumn{5}{|c|}{\textbf{Lower Bounds (Non-Adaptive)}}\\
    \hline
        \cite{mmmw2020hutch_pp} &
        $\Omega(1/\eps)$ & PSD, ram & constant & non-adaptive, $p = 1$\\
    \hline
        \cite{jiang2021optimal}&
        $\Omega(\sqrt{\log(1/\delta)}/\epsilon + \frac{\log(1/\delta)}{\log\log(1/\delta)})$ & PSD, ram &  $\delta$ & non-adaptive, $p = 1$\\
    \hline
     \textbf{This work  \textsuperscript{4}} & $\Omega\left({\log^{p/2} (1/\delta)}/({\epsilon^p(b + \log(1/\epsilon))})  \right)$  & general square, bit & $\delta$ & non-adaptive, general $p$  \\
    
    \hline
     \textbf{This work  \textsuperscript{5}}   & $\Omega\left((\sqrt{\log(1/\delta)}/\eps)^p + \frac{\log(1/\delta)}{\log\log(1/\delta)} \right)$  & general square, ram & $\delta$ & non-adaptive, general $p$\\
    \hline
    
    \end{tabular}}\vspace*{0.3em}
    \caption{Upper and lower bounds on the query complexity for static trace estimation.   In the bit complexity model, each entry of the query vector is specified by $b$ bits, and the dependence on $b$ is necessary. \\
    \textsuperscript{\bf 1}:  A static upper bound generalizing     Hutch++ \cite{mmmw2020hutch_pp} to   Schatten-$p$ norm error  (\cref{thm:hpp-general}).  \\
      \textsuperscript{\bf 2}: An adaptive lower bound via communication complexity of the {Gap Equality} and {Approximate Orthogonality} problem (\cref{thm:main-lowerbound}), which combines \cref{thm:adap-lower} and \cref{thm:adap-lower-2}, resolving an open problem that $\log(1/\delta)$ queries are required in the adaptive setting. \\
       \textsuperscript{\bf 3}: An adaptive lower bound via information-theoretic analysis of Gaussian Wigner matrices (\cref{thm:lb_small}), showing optimal dependence on $\log(1/\delta)$.\\
       \textsuperscript{\bf 4}: A non-adaptive lower bound via communication complexity of {Augmented Indexing} (\cref{thm:non-adaptive}), optimal in all parameters up to the bit complexity term. \\
        \textsuperscript{\bf 5}: A non-adaptive lower bound combining our \cref{thm:lb_small} and the prior result from \cite{jiang2021optimal}.}
    %Our lower bounds show that the additive $\log(1/\delta)$ term in necessary for adaptive $p$-norm trace estimation in the bit model.
    %\\ $\color{red}\star$: holds for small $\eps$; see \cref{thm:lb_small} for a formal statement. 
    \label{tab:bounds_query_complexity}
\end{table}

% \textcolor{red}{Ricahrd: We don't need this since we can combine all lower bounds} 
% \fred{for the gaussian proof, do we still have the $\log(1/\delta)/\log\log(1/\delta)$ lower bound for general $p$, in the way similar to Jiang et al?}
    
% \textcolor{red}{Richard: Actually YES our lower bound holds for any $p >= 1$ (even $p> 2$) but its only non-adaptive. I was thinking the $\log(1/\delta)$ bit complexity adaptive lower bound is stronger perhaps?}

Our work resolves the proposed questions (nearly) optimally, and we next discuss our main results.

\paragraph{Static trace estimation} For Question 1, we prove query complexity lower bounds for implicit trace estimation in both bit complexity and real RAM models of computation, resolving the open problem of establishing unconditional lower bounds for the optimality of Hutchinson's method even in the adaptive setting. 

To do so, we provide new reductions from classic communication complexity problems,  including \textsc{Gap-Equality} and \textsc{Approximate-Orthogonality}, to matrix trace estimation. Our main lower bounds demonstrate that $\log(1/\delta)$ queries are always needed even with adaptivity and for general $p$, there is an additional $1/\epsilon^p$ dependence. A key idea is a communication protocol simulation using the product of two matrices rather than the sum, as was used in prior work on PSD lower bounds \cite{mmmw2020hutch_pp}. 

%We emphasize that while such reductions can easily give lower bounds on $\epsilon$, 
%creating optimal reductions to give lower bounds simultaneously on $\delta$ and $\epsilon$ is non-trivial due to the inherent nature of communication lower bounds, which only requires communication to succeed with constant probability. 
% DW: think the eps by itself is interesting in this paper

\begin{theorem}[Informal; see \cref{thm:main-lowerbound}]
In the bit complexity model, where each entry of each query vector is specified using $b$ bits, 
$$\Omega\left(\frac{1}{\epsilon^p(b + \log(1/\epsilon)) }  + \frac{\log (1/\delta)}{b + \log\log(1/\delta)} \right)$$ 
number of adaptive queries is necessary to achieve $\eps\|\mA\|_p$ error with probability at least $1-\delta$.
\end{theorem}

When adaptivity is not allowed, we give a stronger lower bound (\cref{thm:non-adaptive}) of $$\Omega(\log^{p/2}(1/\delta)/\eps^p).$$
This matches the guarantee of  Hutchinson's non-adaptive estimator up to a constant factor, for which random sign vectors suffice and so one can take $b = O(1)$. 

%\textcolor{red}{Richard: Not sure how to sell this part...do we need this non-adaptive LB anymore?} \fred{The non-adaptive LB is slightly stronger, as the $\log (1/\delta)$ and $1/\eps^2$ terms are multiplied together. I just want to mention it in the intro, and plan to put the corresponding section in appendix for the neurips submission, due to space constraint.}

We also provide a query complexity lower bound in the real RAM model for general Schatten-$p$ norms with $p \in [1,2]$ by using Gaussian ensembles and controlling the remaining entropy of the distribution conditioned on prior queries. In the special case of $p=2$ (i.e., Frobenius norm error guarantee), our bound again matches the classic Hutchinson's method up to a constant factor for $p = 2$, and an additive $\log(1/\delta)$ factor for $p < 2$. Note that in the non-adaptive setting, our lower bound in the RAM model can also be improved for $p < 2$ to include a $\log(1/\delta)/\log\log(1/\delta)$ factor. Therefore, this lower bound emphasizes that our dependence on $\log(1/\delta)$ in the $\epsilon$-dependent term is tight, even in the adaptive setting.

\begin{theorem}[Informal; see \cref{thm:lb_small}]
In the real RAM model, where the queries are real-valued, for sufficiently small $\eps$ and any $p \in [1,2]$,
 $\Omega\left( \left({\sqrt{\log(1/\delta)}}/{\epsilon}\right)^p\right)$
 number of adaptive queries is necessary to  achieve $\eps \|\mA\|_p$ error with probability at least $1-\delta$. 
\end{theorem}
On the algorithmic front, we give a matching upper bound for static trace estimation for general Schatten-$p$ norm error for $p\in [1,2]$. The argument requires a careful balancing of the $\epsilon$ and $\delta$ parameters in the low rank approximation of the Hutch++ procedure from \cite{mmmw2020hutch_pp}. See \cref{thm:hpp-general} for a full statement.

% Fred: I'm commenting the below out, since we are keeping the lower bounds separate. But I think it's a good idea to mention our generalized analysis of Hutch++ in the intro.

% Together, we can optimally characterize the query complexity of static trace estimation for any $\epsilon, \delta > 0$ and $p \in [1,2]$ as

%  \[\Theta\left( \left(\frac{\sqrt{\log(1/\delta)}}{\epsilon}\right)^p + \log(1/\delta)\right)\]

\paragraph{Dynamic trace estimation} 
To answer Question 2, we first give an improved algorithm for dynamic trace estimation that uses a binary tree-based decomposition to estimate all matrix traces with only a small logarithmic overhead. The algorithm improves upon the previous work \cite{dharangutte2021dynamic} and gets an optimal dependence on $0  < \alpha, \delta < 1$, up to logarithmic factors. Specifically, for $p=1$, the prior work gives a method that uses $O\left(m\sqrt{\alpha/\delta}/\eps\right)$ queries for small $\epsilon$, while our algorithm gives an improved $O(m\alpha \sqrt{\log(1/\delta)}/\epsilon)$ bound with a linear dependence on $\alpha$ and square root dependence on $\log(1/\delta)$. For $p=2$, our algorithm matches the query complexity of $O(m\alpha\log(1/\delta)/\eps^2)$ given by previous work. Furthermore, our algorithm works under a general Schatten-$p$ norm  assumption for any $p\in [1,2]$:

\begin{theorem}[Informal; see \cref{thm:alg} and \cref{thm:alg-p}]
For any $p\in [1,2]$, there is a dynamic trace estimation algorithm that achieves error $\eps$ and failure rate $\delta$ at each step.
The algorithm uses a total of 
\begin{equation}\label{eqn:alg-query-p3} 
     \widetilde O \left ((m\alpha+1)  \left ( {\sqrt{\log(1/(\alpha\delta))}}/{\eps}\right)^p + m  \log(1/(\alpha\delta))  \right)
\end{equation}
matrix-vector product queries. Furthermore, for $p = 1$,   it can be improved to
\begin{equation}\label{eqn:alg-query-pequals1} 
     \widetilde O \left ((m\alpha+1)  \left ( {\sqrt{\log(1/(\alpha\delta))}}/{\eps}\right) + m \min(1, \alpha/\epsilon)  \log(1/(\alpha\delta))  \right)
\end{equation}
\end{theorem}
Furthermore, since our algorithm avoids the variance reduction technique from \cite{dharangutte2021dynamic}, we may relax the assumptions of dynamic trace estimation and   require only the first matrix to have   norm $\|\mA_1\|\leq 1$, instead of asking the entire sequence $\mA_i$ to be bounded in such a way. While the norm bound on all $\mA_i$ is crucial for the algorithm in \cite{dharangutte2021dynamic} (re-running the analysis na\"ively would give a worse query complexity of $O(m^3 \alpha^3/\epsilon)$), our tree-based algorithm achieves a nearly optimal query complexity even when the norm of $\mA_i$ grows, and we suffer  only a $\log m$ overhead in that case. Moreover, in our experiments, we find that our algorithm significantly outperforms previous algorithms on real and synthetic datasets. See \cref{sec:experiment} for our experimental results.  

To complement our algorithms, we give  unconditional lower bounds showing that our algorithm is nearly optimal. Our lower bounds rely on a reduction from dynamic trace estimation to static matrix trace estimation from \cite{dharangutte2021dynamic} and make use of our new lower bounds in the static setting.  In particular, the reduction shows that if for a fixed set of parameters $\eps, \delta, p$,  a static trace estimation scheme requires $\Omega(r)$ queries, then $\Omega(m\alpha r)$ queries are necessary for any dynamic algorithm. Combining this observation with our static trace estimation lower bounds, we get:

\begin{theorem}[Informal; see \cref{thm:dy-lb-fro} and \cref{thm:dy-lb-nu}]
For any $p=[1,2)$, our algorithm attains the optimal query complexity, up to bit complexity and logarithmic terms.
\end{theorem}
More specifically, we prove lower bounds that match the first term in our upper bound (\ref{eqn:alg-query-p3}) for all $p \in [1,2]$. For $p=1$, we give a lower bound (\cref{thm:dynamic-hard-instance}) matching the the second term in (\ref{eqn:alg-query-pequals1}) as well, showing that the $m(\log(1/\delta))$ additive dependence is necessary. 

For $p=2$, the prior work \cite{dharangutte2021dynamic} gives a upper bound of $O(m\alpha \log(1/\delta)/\eps^2 + \log(1/\delta)/\eps^2 )$. Our lower bounds are unconditional and show that the first term is tight. Moreover, the second term is necessary due to the static lower bound when $m=1$.   This result is not contradicted by the claim  of  \cref{thm:dynamic-hard-instance}. In particular, when $\alpha \geq  \eps^2$,   \cref{thm:dynamic-hard-instance} is weaker than the   $\Omega (m \alpha \log(1/\delta)/\eps^2)$ lower bound; and when $\alpha <  \eps^2$, the construction by itself requires $\eps/\alpha$ update steps to change the trace by $\eps$, which leads to a lower bound of $\Omega( m \alpha \log(1/\delta)/\eps)$, again weaker than $\Omega (m \alpha \log(1/\delta)/\eps^2)$.

\subsection{Related work}\label{sec:more-work}
 
We summarize prior work on static trace estimation in \cref{tab:bounds_query_complexity}.
The seminal work of \cite{at2011trace_est_psd} gives the first  analysis of Hutchinson's estimator, which was improved by \cite{ra2015implicit_trace_est}.  For PSD matrices, the query complexity can be sharpened, and this was shown recently in  \cite{mmmw2020hutch_pp, jiang2021optimal}. These two papers also give matching lower bounds. The study of dynamic trace estimation was initiated by \cite{dharangutte2021dynamic}, and our work improves upon their results. 

Other applications of implicit trace estimation include inference of Determinantal Point Processes \cite{dong2017log_determinant_gp},  approximating the generalized rank of a matrix \cite{zhang2015distributed}, computing network centrality measures \cite{bergermann2022fast}, matrix
spectrum estimation \cite{han2016spectral_sum,musco2018spectrum}, and eigenvalue counting \cite{di2016efficient}. See \cite{ubaru2018trace_est_applications} for a recent survey.

 \subsection{Organization}
 The remainder of the paper is organized as follows. We give preliminaries in \cref{sec:prelim}. In \cref{sec:algo}, we describe and analyze our improved algorithm for dynamic trace estimation. We study adaptive query lower bounds for static trace estimation in \cref{sec:lower} and show their implications for the dynamic version in \cref{sec:lower-dyn}. 
 Finally, We experimentally validate our algorithm in \cref{sec:experiment}.  
 
\section{Preliminaries}\label{sec:prelim}
 A matrix $\mA \in \R^{n\times n}$ is symmetric positive semi-definite (PSD) if it is real, symmetric and has non-negative eigenvalues. Hence, $x^\top A x \geq 0$ for all $x \in \R^n$. Let $\tr(\mA) = \sum_{i=1}^{n} \mA_{ii}$ denote the trace of $\mA$. Let $\|\mA\|_F = (\sum_{i=1}^{n}\sum_{j=1}^{n} \mA_{ij}^2)^{1/2}$ denote the Frobenius norm and $\|\mA\|_{op} =\sup_{\|\rvv \|_2 = 1}\|\mA \rvv\|_2$ denote the operator norm of $\mA$. We   let $\|\mA\|_p = \left(\sum_i \sigma_i^p\right)^{1/p}$ be the Schatten-$p$ norm, where $\sigma_i$ are the singular values of $\mA$. Two special cases are the Frobenius norm, which   equals   the Schatten-$2$ norm ($\|\mA\|_F = \|\mA\|_2$) and the nuclear norm,  equals   the Schatten-$1$ norm ($\|\mA\|_\star = \|\mA\|_1)$.

% \paragraph{Communication Complexity.} Our lower bound proofs go through communication complexity of the Gap-Hamming and Approximate Orthogonality problem. 
\section{Algorithm for Dynamic Trace Estimation}
\label{sec:algo}
We give an  algorithm for dynamic trace estimation under a general Schatten-$p$ norm assumption, for $p\in[1,2]$.  For $p=1$, our algorithm provides an improved guarantee upon the DeltaShift++ procedure from \cite{dharangutte2021dynamic}. In a later section we complement the result by showing that it is indeed near-optimal. Specifically, we give an algorithm that achieves the following guarantees:
\begin{theorem}[Improved dynamic trace estimation]\label{thm:alg}
 Let $\mA_1,\mA_2,\cdots, \mA_m$ be $n\times n$ matrices such that (1) $\|\mA_i\|_\star \leq 1 $ for all $i$, and (2) $\|\mA_{i+1} -\mA_{i}\|_\star \leq \alpha$ for all $i \leq m-1$. Given matrix-vector multiplication access to the matrices, a failure rate $\delta >0$ and error bound $\eps$, there is an algorithm that outputs a sequence of estimates $t_1,\cdots, t_m$ such that for each $i \in [m]$,
\begin{equation}\label{eqn:dynamic-thm}
    |t_i - \Tr \mA_i| \leq \eps  , \text{ with probability at least }  1-\delta.
    \end{equation}
The algorithm uses a total of 
\begin{equation}\label{eqn:alg-query}
  O \left( (m\alpha+1) \log^2 (1/\alpha)\sqrt{\log(1/
    (\alpha\delta))} /\eps + m \min(1, \alpha/\epsilon)\log(1/(\alpha\delta)) \right)
\end{equation}
matrix-vector multiplication queries to $\mA_1,\mA_2,\cdots, \mA_m$.
\end{theorem}
Compared with DeltaShift++ in \cite{dharangutte2021dynamic}, this guarantee provides  an exponential improvement in $\delta$ and a polynomial improvement in $\alpha$ for $p \neq 2$, while maintaining the optimal dependence on $m$ and $\eps$. 

\subsection{Algorithm}
We  now describe our algorithm.
The first idea is to partition the $m$ updates into groups of size $s = \lceil 1/(2\alpha) \rceil$.  Each group will be treated independently, and we will use 
\begin{equation}\label{eqn:per-group}
     O \left( \log^2 (1/\alpha)\sqrt{\log(1/
    (\alpha\delta))} /\eps + \frac{1}{ \alpha} \log(1/(\alpha\delta)) \right).
\end{equation}
queries on each group. This leads to our claimed query complexity, as there are  $O(m\alpha) $ groups. Note that if $\alpha < \epsilon$,   since $|\Tr(\mA_{j} - \mA_{j-1})| \leq \|\mA_{j} - \mA_{j-1}\|_* \leq \alpha$, the trace can change by at most an additive $\alpha$, so we can simply ignore every subsequence of length $\epsilon/\alpha$. Therefore, we only need to apply our estimators to $m\alpha/\epsilon$ matrices.

Without loss of generality, consider a group of matrices $\mA_1,\cdots, \mA_{1/2\alpha}$. As the first step, we estimate     $ \Tr(\mA_j -\mA_{j-1})$ for each $j\ge 2$ by using the Hutch++ static trace estimator \cite{mmmw2020hutch_pp} as a black box. Then, for each even integer $j=2k$ (for an integer $2\le k \le s/2$), we also estimate  $\Tr (\mA_{2k} - \mA_{2(k-1)})$ in the same way. 
More generally, for each integer  $j = 2^{\ell}k$, for $0 \le \ell < \log_2 s$, we use Hutch++ to approximate     $\Tr\left(\mA_{2^{\ell}k}  -\mA_{2^{\ell}(k-1)}\right)$. 
We view this scheme as     a binary tree: the bottom level consists of leaves corresponding to the  trace  difference of neighboring matrices, and nodes at level $\ell$ correspond to the trace difference of matrices that are $2^\ell$ apart in their indices.

To output an estimate of $\Tr \mA_i$, we will write $i$ in its binary representation and approximate  it by $\Tr(\mA_1)$ plus a sequence of $O(\log(1/\alpha))$ differences, at most one for each level in the binary tree.  By setting the success rates and errors bounds at each level carefully, we can achieve the desired error guarantee of  \cref{eqn:dynamic-thm}.
% For each $ \Tr(\mA_j -\mA_{j-1})$, we invoke \cref{lem:hutch-pp} with parameter $\eps'= \eps/(\alpha \log(1/\alpha))$ and $\delta' = \alpha\delta$.  with parameter $\eps' = \eps/(2^{\ell} \alpha \log(1/\alpha))$ and $\delta' = \alpha\delta$.

To formalize the  construction, we first cite the following guarantee of the Hutch++ algorithm:
\begin{lemma}[Hutch++, nuclear norm, Theorem 5 of \cite{mmmw2020hutch_pp}]\label{lem:hutch-pp}
The Hutch++ estimator  uses $$N=O\left(\sqrt{\log (1 / \delta')} / \varepsilon'+\log (1 / \delta')\right)$$ matrix-vector
multiplication queries such that given any square matrix $\mA$ and parameters $\eps',\delta'$, with probability at least $1-\delta'$,  the algorithm's output $t$ satisfies  
\begin{equation}
    |t - \Tr \mA| \leq \sqrt{\varepsilon'}\left\|\boldsymbol{A}-\boldsymbol{A}_{1 / \varepsilon'}\right\|_{F} \leq \varepsilon'\|\boldsymbol{A}\|_{*}.
\end{equation}
\end{lemma}

Let $ \texttt{Hutch++}(\mA, \eps', \delta')$  denote the output of Hutch++ on matrix $\mA$ with parameters $\eps',\delta'$.  
It will be invoked with different parameters at different levels of the binary tree construction.
A   description of the algorithm is given by the pseudocode  Algorithm \ref{alg:main}, with a helper function   Algorithm  \ref{alg:main-helper}.

For simplicity of analysis, note that since we can add dummy matrices (say, extra copies of $\mA_1$), we assume that each group has size exactly $s=\lceil 1/(2\alpha) \rceil$ and $s$ is a power of two. This  blows up  the total  number of matrices by at most a constant factor.   

\begin{algorithm}[htb]\label{alg:main-bin}
\SetKwInOut{Input}{Input}
\SetKwInOut{Output}{Ouput}

\caption{Improved Dynamic Trace Estimation \label{alg:main}}
  \Input{A sequence of square matrices $(\mA_i)_{i=0}^m \in \mathbb R^{n\times n}$, failure rate $\delta$, error bound $\eps$ \\}
  \Output{Trace estimate $t_i$ for each matrix}
  \vspace*{0.5em}
  Partition the matrices into groups of size $s = \lceil 1/(2\alpha) \rceil$. \\
  For every $g \geq 0$ and $i \in \{0,1,\cdots, s-1\}$, let $\mA^{(g)}_i = \mA_{gs+i+1}$ denote the $i$-th matrix in the $g$-th group.
   
  \For{\text{each group} $\mA^{(g)}_0,\cdots, \mA^{(g)}_{s-1}$ independently}{
  
  % Partition group into size of $u = $
  %% Fred: what's the line above???
  Let $t_0 =  \texttt{Hutch++}(\mA_0^{(g)}, \eps/2,\delta/2))$\\
  \For{each level $\ell$ from $0$ to $\log_2 s -1$}{
  $\textsf{gap} = 2^{\ell}$\\
  \For{$k$ from $1$ to $(s-1)/\textsf{gap}$}{
    Compute  $t_{k,\ell} = \texttt{Hutch++}\left(\mA_{k\cdot\textsf{gap}}  -\mA_{(k-1)\cdot\textsf{gap}}, \eps'(\ell), \delta'\right)$, with   $\eps'(\ell) = \eps/(2^{\ell+1}\alpha \log_2 s)$ and $\delta' = \alpha\delta$.
  }
  }
  Output $t_{gs+i+1} = t_0 + \textsc{SumTree}(1,i,\log_2 s-1, t)$ for each $i \in [0,s-1]$.
 } 
\end{algorithm}

\begin{algorithm}[htb]
\SetKwInOut{Input}{Input}
\caption{\textsc{SumTree}:  Helper Function for Tracing the Binary Tree \label{alg:main-helper} }
  \Input{Indices $i,j$, level $\ell$,  binary tree node values $t$\\}
  \vspace*{0.5em}
  $\textsf{gap} = 2^{\ell}$\\
  \If{$j\leq i$}{\Return $0$.}
  \If{\textsf{gap} $= 1$}{\Return $t_{\ell,i}$.}
  \If{$j-i \geq \textsf{gap}$}{\Return $t_{\ell, \lfloor (j-1)/\textsf{gap} \rfloor} + \textsc{SumTree}(i+\textsf{gap},j,\ell-1,t)$.}
  \Else{\Return $\textsc{SumTree}(i,j,\ell-1,\textsf{gap})$.}
\end{algorithm}

\subsection{Analysis}
The analysis of the algorithm is rather lengthy and is delayed to \cref{sec:proof-apx-alg}. 
In addition, we give a general analysis of the algorithm under Schatten-$p$ norm assumption and the specific improved bounds for $p = 1$ in \cref{sec:gen-alg-apx} and show how to relax the bounded norm assumption in \cref{sec:relax-apx}.

\section{Lower Bounds for Adaptive Trace Estimation}\label{sec:lower}
In this section, we  provide  (nearly) optimal lower bounds for    trace estimation with adaptive matrix-vector multiplication queries, under general square matrices and Schatten-$p$ norm error.   

\subsection{Adaptive Lower Bound, Bit Complexity}
First, we show two separate lower bounds under bit complexity model, both proven via reductions from communication complexity problems. One shows an $\Omega(1/\eps^p)$ lower bound (\cref{thm:adap-lower}) and the other $\Omega(\log(1/\delta))$ (\cref{thm:adap-lower-2}), up to bit complexity terms. 
Combined together, they yield:
\begin{theorem}[Adaptive query lower bound, bit complexity]\label{thm:main-lowerbound}
 Any     algorithm  that accesses a   square matrix $\mA$ via matrix-vector multiplication queries  requires  at least $$\Omega\left(\frac{1}{\epsilon^p(k+\log(1/\epsilon)) }  + \frac{\log (1/\delta)}{k+\log\log(1/\delta)} \right)$$   queries  to output  an estimate $t$ such that with probability at least $1-\delta$, $|t - \Tr\mA| \leq \epsilon \|\mA\|_p$,   for any $p\in[1,2]$, where the query vectors may be adaptively chosen with  entries    specified by $k$ bits.
 \end{theorem}
The proofs of the theorems can be found in \cref{sec:adap-lower-pf}.
 
 \subsection{Adaptive Lower Bound, RAM}

Next, we prove a tight lower bound under the real RAM model (\cref{thm:lb_small}). The bounds hold for any Schatten-$p$ norm error. 
Our proof is via information-theoretic analysis of random Gaussian matrices and is delayed to \cref{sec:pf-lb-small}. 
\begin{theorem}[Lower Bound for Any Schatten Norm] 
\label{thm:lb_small}
    For all $p \in [1,2]$, $ \delta > 0$ and $0 < \epsilon < (\log(1/\delta))^{1/2 - 1/p}$, any algorithm that takes in any input matrix $\mA$ and succeeds with probability at least $1-\delta$ in outputting an estimate $t$ such that $| t - \tr(\mA)| \leq \epsilon \|\mA\|_p$  requires $$m =  \Omega\left( \left(\frac{\sqrt{\log(1/\delta)}}{\epsilon}\right)^p\right) $$ matrix-vector multiplication queries.
\end{theorem}

\section{Lower Bounds for Dynamic Trace Estimation}\label{sec:lower-dyn}
Using the query complexity lower bounds for adaptive trace estimation, we can now prove tight lower bounds for dynamic trace estimation. The recent work of Dharangutte and Musco \cite{dharangutte2021dynamic} only provides a \textit{conditional} lower bound, assuming that   Hutchinson's scheme is  optimal. We  remove this assumption and  make the lower bound unconditional. 
We additionally prove a lower bound by constructing an explicit hard instance in the dynamic setting. Our lower bounds hold     under a general Shatten norm assumption and nearly matches the guarantee of our algorithm.

\subsection{Lower Bounds via Static-to-Dynamic Reduction}
 We first show a   lower bound for    dynamic trace estimation under a Frobenius norm assumption. This immediately implies that the {DeltShift} algorithm due to \cite{dharangutte2021dynamic} is optimal for $p=2$. 
 
First, we cite a  static-to-dynamic reduction from \cite{dharangutte2021dynamic} and its implication.  The reduction shows how to solve a static instance using a dynamic trace estimation scheme, and therefore any hardness on the static problem translates to the dynamic setting as well. It holds generally for an error bound in any Schatten norm. For completeness, we give a proof in \cref{sec:cond-pf}.

\begin{lemma}[Conditional lower bound for dynamic trace estimation \cite{dharangutte2021dynamic}]\label{lem:cond-lb}
Suppose that any algorithm  that achieves \cref{eqn:static-p} for static trace estimation  must use $\Omega (r)$ matrix-vector product queries. Then any dynamic trace estimation algorithm requires $\Omega(r \alpha m)$ matrix-vector product queries under a general Schatten-$p$ norm assumption, when $\alpha = 1/(m-1)$. 
\end{lemma}
It follows immediately from this lemma and our adaptive query lower bound  (\autoref{thm:main-lowerbound}):
\begin{theorem}[Unconditional lower bound for dynamic trace estimation, bit]\label{thm:dy-lb-fro}
For all $p\in [1,2]$ and $\eps,\delta\in(0,1)$, any algorithm for dynamic trace estimation under a Schatten-$p$ norm assumption must use at least $$\Omega\left(\alpha m\left(\frac{1}{\epsilon^p(k+\log(1/\epsilon)) }  + \frac{\log (1/\delta)}{k+\log\log(1/\delta)} \right)\right)$$   matrix-vector multiplication queries, where each entry  of the query vectors is specified by $k$ bits.
\end{theorem}
% Note that in the regime of constant success probability and $p=2$, the lower bound  precisely matches the guarantee of the DeltaShift algorithm in \cite{dharangutte2021dynamic} (up to bit complexity).

Combining the same reduction (\cref{lem:cond-lb}) with our previous real RAM lower bound (\cref{thm:lb_small}) in the static setting gives:
\begin{theorem}[Unconditional lower bound for dynamic trace estimation, RAM] \label{thm:dy-lb-nu}
For all $p \in [1,2]$, $ \delta > 0$ and $0 < \epsilon < (\log(1/\delta))^{1/2 - 1/p}$, any algorithm for dynamic trace estimation  under a Schatten-$p$ norm assumption must use at least $  \Omega\left(\alpha m \left({\sqrt{\log(1/\delta)}}/{\epsilon}\right)^p\right)$   matrix-vector multiplication queries.
\end{theorem}

\subsection{Lower Bound via Explicit Hard Instance}

Using the hard instance based on \textsc{Gap-Equality} in the static setting (from the proof of \cref{thm:adap-lower-2}), we give  an explicit  hardness construction against any dynamic trace estimation scheme. This yields the following lower bound, and its proof is in \cref{sec:pf-dynamic-hard-instance}.
\begin{theorem}\label{thm:dynamic-hard-instance}
For all $p\in [1,2]$ and   $\eps ,\delta \in (0,1/4)$, any algorithm for dynamic trace estimation  under Schatten-$p$ norm assumption must use at least $$  \Omega\left(m \min\left(1, \frac{\alpha}{\epsilon}\right) \frac{\log(1/\delta)}{k+ \log\log(1/\delta)} \right)$$   matrix-vector multiplication queries, where each entry of the query vectors is specified by $k$ bits. 
\end{theorem}
\section{Experiments} \label{sec:experiment}
We experimentally validate  our algorithmic results. We compare  Algorithm \ref{alg:main},  with the following procedures on both synthetic and real datasets.
 More experimental details are  in    \cref{sec:exp_details}. 
\begin{itemize}
    \item Hutchinson's: Apply the classic Hutchinson's scheme for each $\Tr(A_i)$ independently.
    \item DiffSum: Approximate $t_1\approx \Tr(\mA_1)$ and each neighboring  difference $d_i \approx \Tr(\mA_i) - \Tr(\mA_{i-1})$ using Hutchinson's independently. Then output $t_i = t_1 + \sum_{j=2}^i d_j$.
    \item DeltaShift:  The main algorithm of \cite{dharangutte2021dynamic}. The experiments from \cite{dharangutte2021dynamic} demonstrate that DeltaShift outperforms DiffSum and other Hutchinson-based schemes on various datasets.
\end{itemize}

\paragraph{Synthetic data.} We simulate a dynamic trace estimation instance by first generating a (symmetric) random matrix $\mA^{n\times n}$ and then adding random perturbations over $T=100$ time steps. The    details and results  are found in \cref{sec:exp-synth}.

\paragraph{Counting triangles.} Our first  experiment  on a real-world dataset is on counting triangles in  dynamic undirected (simple) graphs. Note that the number of triangles in a graph equals $\frac{1}{6} \Tr \mA^3$, where $\mA$ is the adjacency matrix of the graph. Thus, triangle counting reduces to trace estimation.

We use two arXiv collaboration networks with $5,242$ and $9,877$ nodes \cite{leskovec2007graph}.\footnote{The first   is the collaboration network of arXiv General Relativity (ca-GrQc) and the second   High Energy Physics Theory (ca-HepTh). Both   are available at \url{https://sparse.tamu.edu/SNAP}.} The nodes represent authors, and edges indicate co-authorships. To simulate a real-world scenario,  we add a random clique of size at most $6$ to the graph in each step, indicating a   group of researchers jointly publishing a paper. We note that our algorithm significantly outperforms other methods (\cref{fig:arxiv}). 

\begin{figure}
\centering
\begin{subfigure}{.5\textwidth}
  \centering
  \includegraphics[scale=0.4]{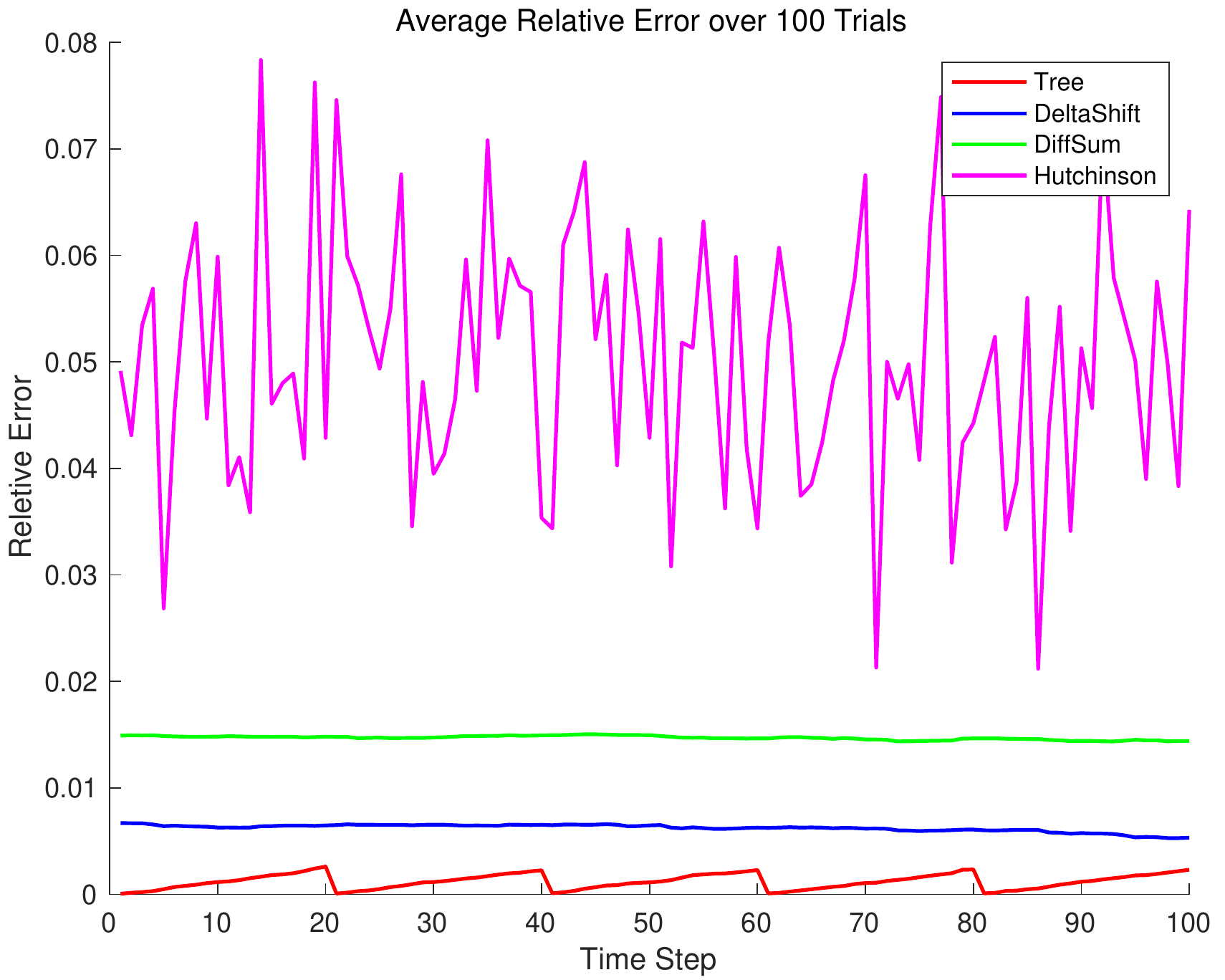}
  \caption{SNAP/ca-GrQc}
  \label{fig:grqc}
\end{subfigure}%
\begin{subfigure}{.5\textwidth}
  \centering
  \includegraphics[scale=0.4]{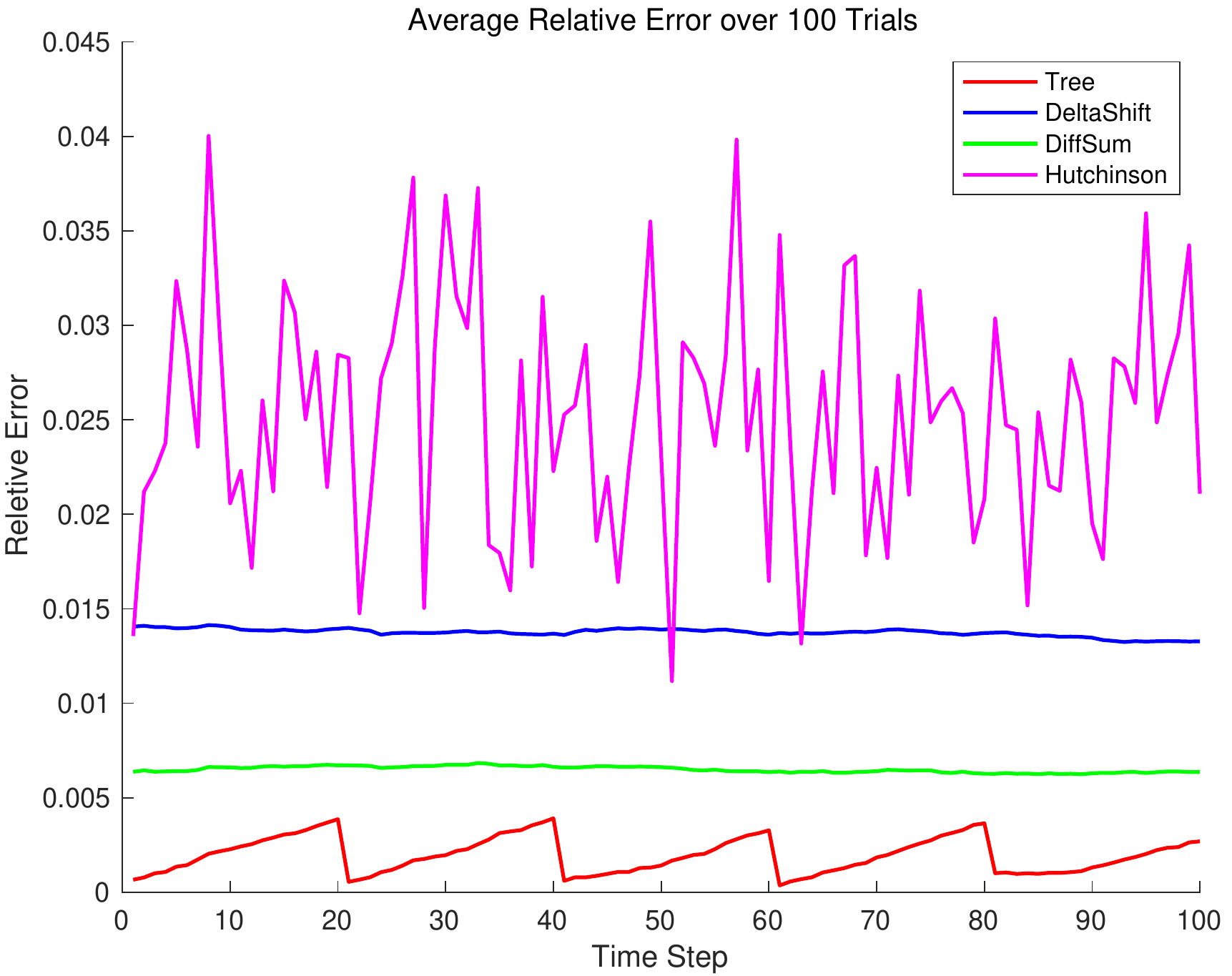}
  \caption{SNAP/ca-HepTh}
  \label{fig:hepth}
\end{subfigure}
\caption{ArXiv datasets. Query budget is $8,000$. In this experiment, the trace values are large, so we measure the performance of the algorithms by their relative   error $|t_i  - \Tr \mA_i^3|  / \max_i {\Tr \mA_i^3}$. }
\label{fig:arxiv}
\end{figure}

\paragraph{Neural network weight matrix.}
\begin{figure}
\centering
\begin{subfigure}{.5\textwidth}
  \centering
  \includegraphics[scale=0.4]{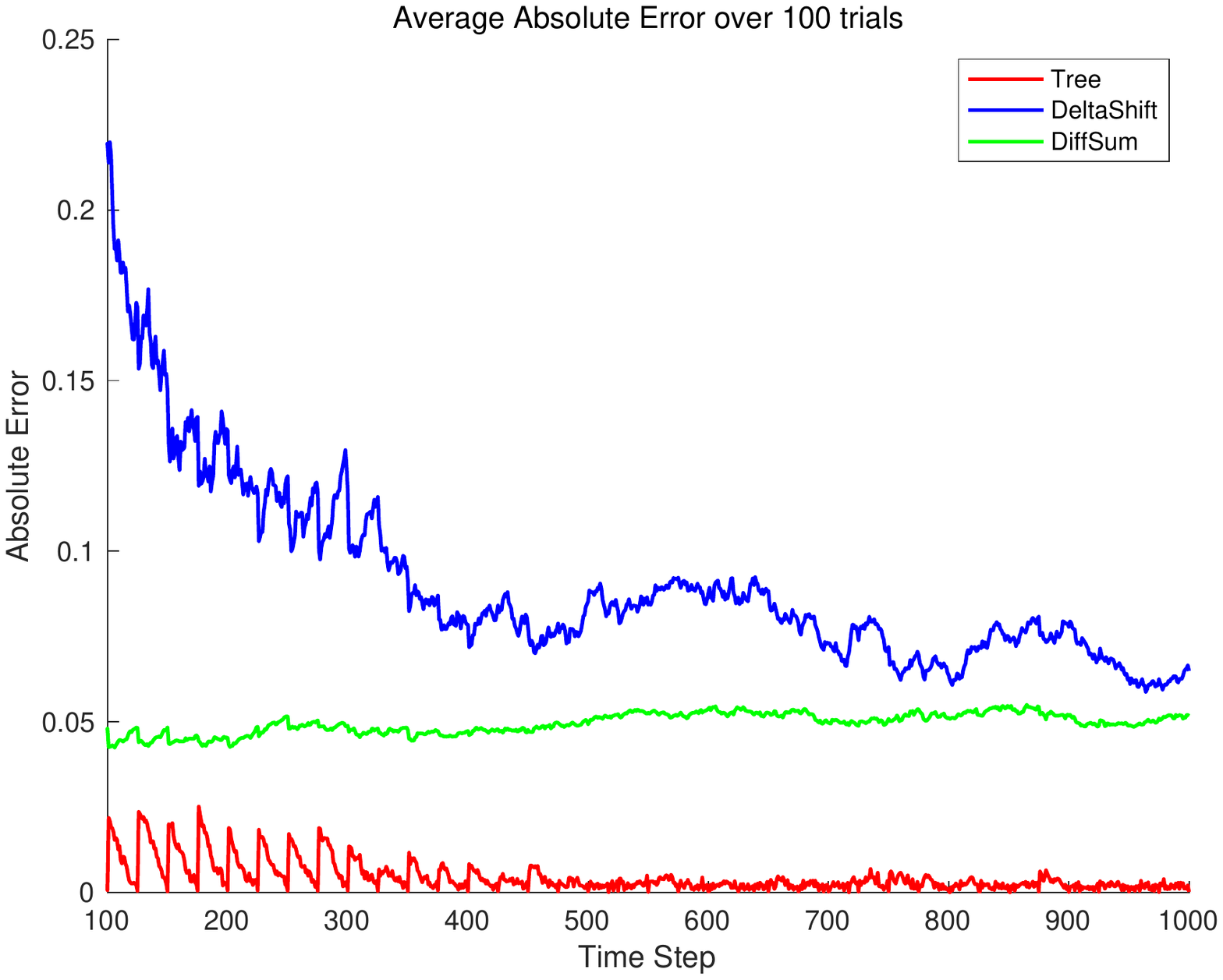}
  \caption{Error over last $900$ time steps}
  \label{fig:mnist}
\end{subfigure}%
\begin{subfigure}{.5\textwidth}
  \centering
  \includegraphics[scale=0.4]{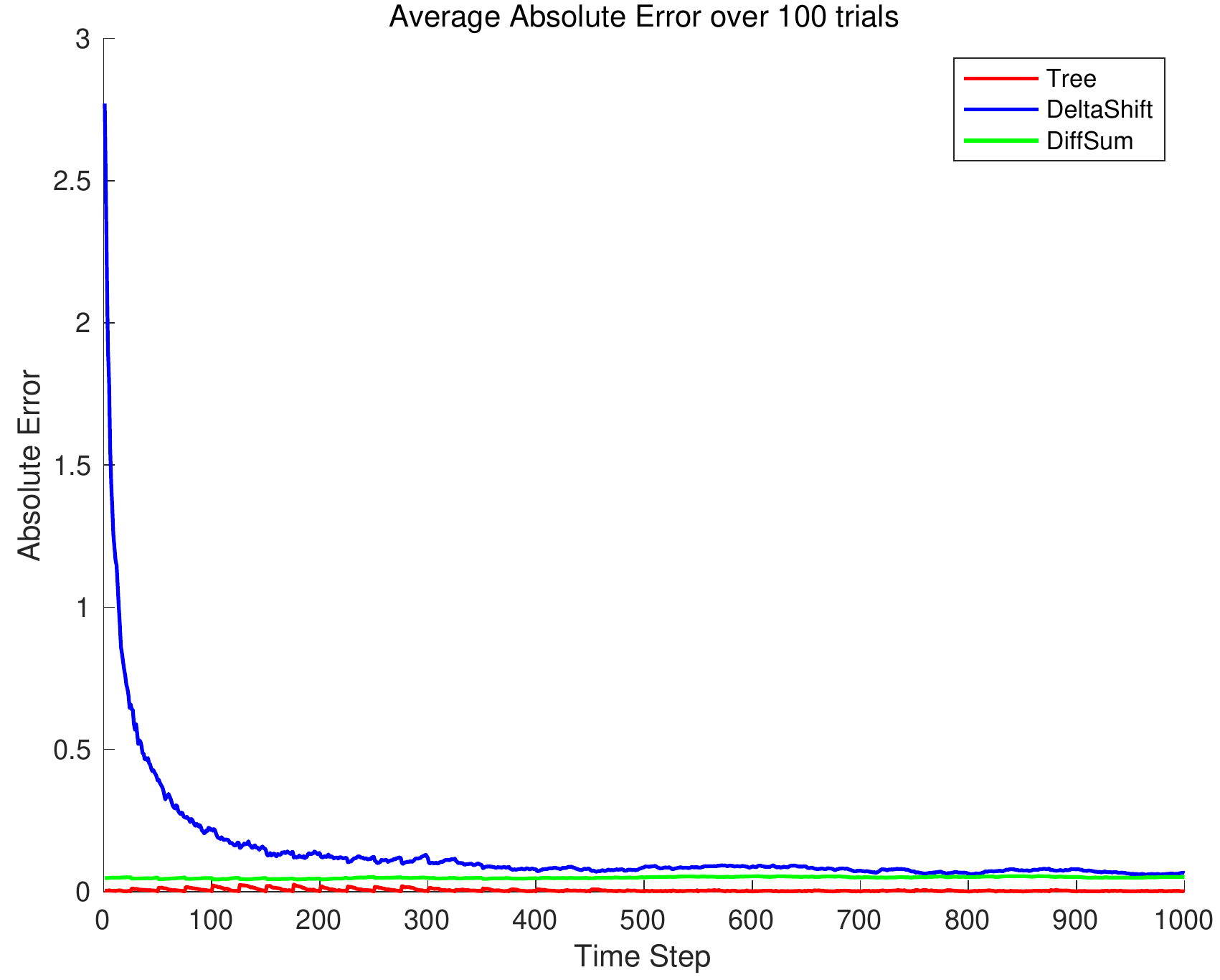}
  \caption{Error over entire $1,000$ time steps}
  \label{fig:mnist2}
\end{subfigure}
\caption{MNIST. Query budget is $50,000$.}
\label{fig:mnist-double}
\end{figure}

We evaluate the performance of the algorithms on a sequence of weight matrices of a neural network, generated during the training process. In particular, we choose a three-layer neural network with a hidden layer of $100\times 100$. We train the network on the MNIST dataset via mini-batch SGD and consider the first $1,000$ steps, when the weights  are changing most rapidly. Our algorithm achieves much smaller error than DiffSum and DeltaShift (\cref{fig:mnist-double}).

\bibliographystyle{alpha}
\bibliography{bib}

%%%%%%%%%%%%%%%%%%%%%%%%%%%%%%%%%%%%%%%%%%%%%%%%%%%%%%%%%%%%
 
%%%%%%%%%%%%%%%%%%%%%%%%%%%%%%%%%%%%%%%%%%%%%%%%%%%%%%%%%%%%
\newpage

\appendix
\section{Background on Communication Complexity}
Our lower bound proofs use communication complexity. In a communication problem, Alice and Bob receive inputs $x\in \{-1,1\}^m$ and $y \in \{-1,1\}^m$, repsectively, and wish to compute a function $f: \{-1,1\}^m \times \{-1,1\}^m \rightarrow \{-1,1\}$. The players communicate according to a protocol $P$ and end with an agreed-upon value $z$. The sequence of binary messages exchanged by the players is called the transcript of $P$, denoted $P(x,y)$. We say that the protocol computes $f$ with error $\delta$ if $\Pr(z \neq f(x,y) ) \leq \delta$. Let $\text{CC}(P)$ be the length (in bits) of the transcript $P(x,y)$. The communication complexity of $f$ is defined to be the minimum communication cost of any protocol with error $\delta$:
\begin{equation}
    \text{CC}_\delta(f) = \min \{ \text{CC}(P): P \text{ computes } f \text{ with error } \delta\}.
\end{equation}
\section{Proof Details of \texorpdfstring{\cref{sec:algo}}{Section 3}}\label{sec:sec2-pr}
\subsection{Proof of \texorpdfstring{\cref{thm:alg}}{Theorem 3.1}} \label{sec:proof-apx-alg}

To give a proof sketch, we consider a fixed group and a constant $\delta = \Theta(1)$. To bound the query complexity, we observe that within the group, each level of the binary tree incurs roughly the same number of matrix-vector product queries. 
Moreover, at the bottom level, there are $s$ calls (including the one computing $t_0$) to \texttt{Hutch++}, with $\eps' = \widetilde O(\eps/\alpha)$ and $\delta' = O(\alpha)$. By \cref{lem:hutch-pp}, each call uses $\widetilde O(\alpha /\eps)$ queries.  Hence, each level uses $\widetilde O(s\alpha /\eps)= \widetilde O(1/\eps)$ queries. Since there are $\log (1/\alpha)$ levels per group and $O(m\alpha)$ groups, this gives a bound of $\widetilde O(m\alpha / \eps)$ on the total number of queries, as claimed in \cref{eqn:alg-query}. A similar argument shows that the scheme achieves the desired error bound $\eps$ and failure probability $\delta$. We  formally prove \cref{thm:alg}:
\begin{proof}[Proof of \cref{thm:alg}]
Fix a group  $g$ and an index $i$. We first argue that the output $t_{gs+i+1}$ is an accurate estimate of the trace $\Tr (\mA_i^{(g)})$, namely, one which satisfies \cref{eqn:dynamic-thm}. By construction, the \textsc{SumTree} algorithm decomposes the $t_{gs+i+1} - t_0$ into at most $\log_2 s$ terms, one at each level $\ell$. Each term is an estimate $t_{k,\ell}$ of $\Tr\left(\mA^{(g)}_{2^\ell k} - \mA^{(g)}_{2^\ell(k-1)}\right)$, for some $\ell,k$.   By assumption, each increment  $\mA^{(g)}_i -\mA^{(g)}_{i-1}$ has Schatten-$1$ norm at most $\alpha$. Hence, by the triangle inequality, 
\begin{align*}
      \left\|\mA^{(g)}_{2^\ell k} -   \mA^{(g)}_{2^\ell(k-1)}\right\|_\star  = \left\| \sum_{j =2^\ell k +1}^{2^\ell (k-1)} \mA^{(g)}_{j} - \mA^{(g)}_{j-1} \right\|_\star \le 2^{\ell}  \alpha
\end{align*}
 By the guarantee of the Hutch++ estimator (\cref{lem:hutch-pp}) and the inequality above, we have that for all $\ell, k$
\begin{align}
    \left| t_{k, \ell} - \Tr \left(\mA^{(g)}_{2^\ell k} -  \mA^{(g)}_{2^\ell(k-1)}\right) \right| &\le \epsilon'(\ell)   \left\|\mA^{(g)}_{2^\ell k} -   \mA^{(g)}_{2^\ell(k-1)}\right\|_\star \nonumber   \\&\leq \epsilon'(\ell)  \cdot 2^\ell \alpha\nonumber  \\&= \eps/ (2\log_2 s)\label{eqn:success-one-call},
\end{align}
with probability at least $1-\delta'$. 
Again, by the guarantee of Hutch++, $t_0$ approximates $\Tr \mA^{(g)}_0$ up to an $(\eps/2)\|\mA^{(g)}_0\|_\star$ additive error. Therefore,  conditioned on \cref{eqn:success-one-call},  the total error of the estimate  $t_{gs+i+1}$ for all $g,i$ is bounded by 
\begin{align}
    \left| t_{gs+i+1} - \Tr\left(\mA^{(g)}_i\right)\right| &\leq   (\eps/2) \left\|\mA^{(g)}_0\right\|_\star+  (\log_2  s ) \cdot  \eps/ (2\log_2 s) \nonumber \\
    &\leq (\eps/2 )\left\|\mA^{(g)}_0\right\|_\star+ \eps/2\nonumber \\
    &\leq \eps\label{eqn:success-all}
\end{align}
where the first line follows since there are $\log_2 s$ levels and the last line since $\|\mA_0^{(g)}\|_\star \leq 1$ by assumption. To bound the failure rate, we note that for a fixed $g$ and $i$, \cref{eqn:success-all} holds if \cref{eqn:success-one-call} holds for all $t_{ k, \ell}$ that are accessed in computing $t_{gs+i+1}$  (via the \textsc{SumTree} procedure).  By construction, there are at most $\log_2 s$ of these terms, where the bound follows from the number of levels of the binary tree.  A simple union bound yields the desired guarantee \cref{eqn:dynamic-thm}.

It remains to prove the bound on the query complexity (\cref{eqn:alg-query}). 
Each group is treated identically, so we   consider any fixed group. 
 Within each group (of size $s$) and at each level $\ell$, we make $O(s/2^\ell )$ calls to \texttt{Hutch++}$(\mA, \eps'(\ell),\delta')$. 
 By \cref{lem:hutch-pp}, this leads to
\begin{equation*}
    O\left( \left(s/ 2^{\ell} \right) \cdot\left( \sqrt{\log (1 / \delta')} / \varepsilon'(\ell) +\log (1 / \delta') \right) \right)
\end{equation*}
queries at level $\ell$. Plugging in the values of $\eps'(\ell),\delta'$, this equals  
\begin{equation*}
       O \left( \log (1/\alpha)\sqrt{\log(1/
    (\alpha\delta))} /\eps + \frac{1}{2^{\ell} \alpha} \log(1/(\alpha\delta)) \right).
\end{equation*}
Summing over $\ell \in \{0,1, \cdots,\log_2 s-1\}$, where $s = O(1/\alpha)$, we have that within each group, the number of queries is bounded by
\begin{equation*}
    O \left( \log^2 (1/\alpha)\sqrt{\log(1/
    (\alpha\delta))} /\eps + \frac{1}{ \alpha} \log(1/(\alpha\delta)) \right).
\end{equation*}
There is a total of $\max \{ 1,O(m\alpha)\}$ groups. Also, we consider the case when $\alpha < \eps$, where the algorithm only needs to provide a fresh estimate every $\eps/\alpha$ time steps. Hence, the sequence length is reduced effectively to $m\alpha/\eps$. Therefore, the   query complexity of Algorithm \ref{alg:main} is at most
\begin{equation*}
    O \left( (m\alpha+1)\log^2 (1/\alpha)\sqrt{\log(1/
    (\alpha\delta))} /\eps + m \min\{1,\alpha/\eps\} \log(1/(\alpha\delta)) \right).
\end{equation*}
This finishes the proof.
\end{proof}

\subsection{General Schatten-\textit{p} Norm Analysis}\label{sec:gen-alg-apx}

%Should follow directly from the same tricks in \cite{mmmw2020hutch_pp} for relating the trace/Frobenius norm with Schatten-$p$ norm.

To generalize trace estimation to $\epsilon \|\mA\|_p$ error for any $p \in [1,2]$, we need to revisit the variance reduction technique to  achieve $O(1/\epsilon)$ query complexity for the nuclear norm. The technique rewrites $\mA = \mB_k + \Delta_k$, where $\mB_k$ is a rank-$k$ matrix with a determined trace, and $\|\Delta_k\|_F \leq O(1)\|\mA - \mA_k\|_F$, where $\mA_k$ is the best rank-$k$ approximation to $\mA$ for some $k$. Then, we can approximate $\tr(\mA)$ by first explicitly calculating $\tr(\mB_k)$ and then approximating the trace of $\Delta_k$. The two step procedure requires a careful balancing for how queries are spent between the two components to minimize the total estimation error in the Schatten $p$ norm and results in a $O(1/\epsilon^p)$ query complexity.

\begin{theorem}[general Schatten-$p$ error analysis of Hutch++]\label{thm:hpp-general}
The Hutch++ estimator of rank $k$ generalizes to a matrix $\mA$ with any Schatten norm $p$ bound for $p \in [1,2]$, and satisfies $|\tr(\mA) - \texttt{Hutch++}(\mA)| \leq \epsilon \|\mA\|_p$ with a total matrix-vector query complexity of  

$$O\left(\left( \frac{\sqrt{\log(1/\delta)}}{\epsilon}\right)^p + \log(1/\delta) \right).$$

\end{theorem}
\begin{proof}[Proof of \cref{thm:hpp-general}]
Let $\mA_k$ be the best rank-$k$ approximation to $\mA$. Then the $\texttt{Hutch++}$ estimator allows us to estimate the trace of $\mA$ by writing $\mA = \mA_k + \Delta$, where $\|\Delta\|_F \leq 2 \|\mA - \mA_k\|_F$. Then, we can directly calculate $\Tr(\mA_k)$ and use Hutchinson's method \cite{hut1990hutchisons_method} with $\ell$ matrix-vector multiplication queries, which gives a standard additive error guarantee of 
\[ C\sqrt{\frac{\log(1/\delta)}{\ell}} \|\Delta\|_F,\]
for some fixed constant $C$. Now, we use the fact that if $\sigma_i$ are the singular values of $\mA$, then by the definition of Schatten norms,
\[\|\Delta\|_F \leq 2 \|\mA - \mA_k\|_F \leq \sqrt{\sum_{i=k}^{n} \sigma_i^2} \leq  \sqrt{\sigma_k^{2-p} \sum_{i=k}^n \sigma_i^p}.\]

Note that we have the following inequality: $k \sigma_k^p \leq \|\mA\|_p^p$. Therefore, rearranging gives $\sigma_k^{2-p} \leq k^{1 - 2/p} \|\mA\|_p^{2-p}$. Finally, we conclude that the total error is bounded by
\begin{align*}
   \|\tr(\mA) - \texttt{Hutch++}(\mA) \| &\leq C \sqrt{\frac{\log(1/\delta)}{\ell}} \|\Delta\|_F \\
   &\leq C \sqrt{\frac{\log(1/\delta)}{\ell}} \sqrt{\sigma_k^{2-p}\sum_{i=k}^n \sigma_i^p} \\
   &\leq C \sqrt{\frac{\log(1/\delta)}{\ell}} \sqrt{k^{1-2/p}\|\mA\|_p^{2-p}\|\mA\|_p^p} \\
   &\leq C \sqrt{\frac{\log(1/\delta)}{\ell k^{2/p-1}}} \|\mA\|_p. 
\end{align*}
Since we want to set $k = \ell$ to minimize the query complexity, it follows that to reduce to error to $\epsilon $, we need a number of queries equal to:
$$k =\ell =  \left( \frac{\sqrt{\log(1/\delta)}}{\epsilon}\right)^p. $$

Finally, by the same analysis of \cite{mmmw2020hutch_pp}, using $O(k + \log(1/\delta))$ matrix-vector products suffices to obtain a constant-factor rank-$k$ approximation of $\mA$, and further, we want $l \geq \log(1/\delta)$. This concludes the proof. 
\end{proof} 
With this generalized analysis of Hutch++, one can  easily extend \cref{thm:alg} and obtain:
\begin{theorem}[general Schatten-$p$ norm analysis of Algorithm \ref{alg:main}]
\label{thm:alg-p}
 Let $\mA_1,\mA_2,\cdots, \mA_m$ be $n\times n$ be matrices such that (1) $\|\mA_i\|_p \leq 1 $ for all $i$, and (2) $\|\mA_{i+1} -\mA_{i}\|_p \leq \alpha$ for all $i \leq m-1$ and some $p\in[1,2]$. Given matrix-vector multiplication access to the matrices, a failure rate $\delta >0$, and an error bound $\eps > 0$, there is an algorithm that outputs a sequence of estimates $t_1,\cdots, t_m$ such that for each $i \in [m]$,
\begin{equation}\label{eqn:dynamic-thm-p2}
    |t_i - \Tr \mA_i| \leq \eps  , \text{ with probability at least }  1-\delta.
    \end{equation}
The algorithm uses a total of 
\begin{equation}\label{eqn:alg-query-p2}
  O \left ((m\alpha+1) \log^{1+p}(1/\alpha)\left ( \frac{\sqrt{\log(1/(\alpha\delta))}}{\eps}\right)^p + m \log(1/(\alpha\delta))  \right).
\end{equation}
matrix-vector multiplication queries to $\mA_1,\mA_2,\cdots, \mA_m$.
\end{theorem}
The proof is via counting the number of queries differently using the general Hutch++ analysis \cref{thm:hpp-general}. 
\begin{proof}[Proof of \cref{thm:alg-p}]
 Our error analysis is almost identical to the nuclear norm case. We sketch it here for completeness. Consider any fixed group $g$.
 By assumption every increment  $\mA^{(g)}_i -\mA^{(g)}_{i-1}$ has Schatten-$p$ norm at most $\alpha$.  Hence, by the triangle inequality, 
\begin{align*}
      \left\|\mA^{(g)}_{2^\ell k} -   \mA^{(g)}_{2^\ell(k-1)}\right\|_p  = \left\| \sum_{j =2^\ell k +1}^{2^\ell (k-1)} \mA^{(g)}_{j} - \mA^{(g)}_{j-1} \right\|_p \le 2^{\ell}  \alpha
\end{align*}
 By the general Schatten-$p$ norm analysis of the Hutch++ estimator (\cref{thm:hpp-general}) and the inequality above, we get that for all $\ell, k$:
\begin{align}
    \left| t_{k, \ell} - \Tr \left(\mA^{(g)}_{2^\ell k} -  \mA^{(g)}_{2^\ell(k-1)}\right) \right| &\le \epsilon'(\ell)   \left\|\mA^{(g)}_{2^\ell k} -   \mA^{(g)}_{2^\ell(k-1)}\right\|_p \nonumber   \\&\leq \epsilon'(\ell)  \cdot 2^\ell \alpha\nonumber  \\&= \eps/ (2\log_2 s)\label{eqn:success-one-call-p},
\end{align}
with probability at least $1-\delta'$. 
Again, by the guarantees of Hutch++, $t_0$ approximates $\Tr \mA^{(g)}_0$ up to an $(\eps/2)\|\mA^{(g)}_0\|_p$ additive error. Therefore, conditioned on \cref{eqn:success-one-call-p}, the total error of the estimate $t_{gs+i+1}$ for all $g,i$ is bounded by 
\begin{align}
    \left| t_{gs+i+1} - \Tr\left(\mA^{(g)}_i\right)\right| &\leq   (\eps/2) \left\|\mA^{(g)}_0\right\|_p+  (\log_2  s ) \cdot  \eps/ (2\log_2 s) \nonumber \\
    &\leq (\eps/2 )\left\|\mA^{(g)}_0\right\|_p+ \eps/2\nonumber \\
    &\leq \eps\label{eqn:success-all-p}.
\end{align}
A union bound thus proves the accuracy guarantee (\autoref{eqn:dynamic-thm-p2}).

We now count the query complexity differently using \cref{thm:hpp-general}.
As before, within each group (of size $s$) and in each level $\ell$, we make $O(s/2^\ell )$ calls to \texttt{Hutch++}$(\mA, \eps'(\ell),\delta')$.
This leads to a total number of 
\begin{equation}
    O\left( \left(s/2^{\ell}\right) \cdot\left ( \left( \frac{\sqrt{\log(1/\delta')}}{\epsilon'(\ell)}\right)^p + \log(1/\delta')
    \right)\right)
\end{equation}
matrix-vector multiplication queries by \cref{thm:hpp-general}.  Substituting $\eps'(\ell)  = \eps / 2^{\ell+1}\alpha \log_2 s$ and $\delta' = \alpha \delta$, we have
\begin{equation*}
    O \left ( \log^p(1/\alpha)\left ( \frac{\sqrt{\log(1/(\alpha\delta))}}{\eps}\right)^p + + \frac{1}{2^{\ell} \alpha} \log(1/(\alpha\delta))  \right).
\end{equation*}

Note that we have used the fact that $\alpha 2^{l+1} \leq \alpha s \leq 1$ to simplify the expression. Summing over $\ell \in \{0,1, \cdots,\log_2 s-1\}$, where $s = O(1/\alpha)$, we obtain that within each group, the number of queries is bounded  by
\begin{equation*}
    O \left ( \log^{1+p}(1/\alpha)\left ( \frac{\sqrt{\log(1/(\alpha\delta))}}{\eps}\right)^p + \frac{1}{  \alpha} \log(1/(\alpha\delta))  \right).
\end{equation*}
Since there are $\max \{ 1, O(m\alpha)\}$ groups, the total query complexity is at most 
\begin{equation*}
    O \left ((m\alpha+1) \log^{1+p}(1/\alpha)\left ( \frac{\sqrt{\log(1/(\alpha\delta))}}{\eps}\right)^p + m \log(1/(\alpha\delta))  \right).
\end{equation*}
This completes the proof.
 \end{proof}

\subsection{Relaxing Assumptions}\label{sec:relax-apx}

Recall that for dynamic trace estimation, we generally require all matrices $\mA_i$ to have unit-bounded Schatten-$p$ norm. While it is often the case that the initial matrix $\mA_1$ has controlled norm, in practice it is unrealistic to assume a general bound on the matrix norm upon dynamic updates. Of course, note that due to the bounded difference assumption, we can always use a linear bound $\|\mA_i\|_p \leq 1 + \alpha i$. However, using this bound na\"ively with the analysis of other algorithms, such as Hutchinson's or its variance-reduced version of \cite{dharangutte2021dynamic}, introduces additional $\text{poly}(m, \alpha)$ terms in the query complexity. Instead, we show that our tree-based procedure without any initial partitioning still attains an optimal dependence on $m$ and $\alpha$ for the nuclear norm.

\begin{theorem}[general Schatten-$p$ norm analysis of non-partitioned Algorithm \ref{alg:main}]
\label{thm:alg-p-relaxed}
 Let $\mA_1,\mA_2,\cdots, \mA_m$ be $n\times n$ matrices such that (1) $\|\mA_1\|_* \leq 1 $ and (2) $\|\mA_{i+1} -\mA_{i}\|_* \leq \alpha$ for all $i \leq m-1$. Given matrix-vector multiplication access to the matrices, a failure rate $\delta >0$ and error bound $\eps$, there is an algorithm that outputs a sequence of estimates $t_1,\cdots, t_m$ such that for each $i \in [m]$,
\begin{equation}\label{eqn:dynamic-thm-p}
    |t_i - \Tr \mA_i| \leq \eps  , \text{ with probability at least }  1-\delta.
    \end{equation}
The algorithm uses a total of 
\begin{equation}\label{eqn:alg-query-p}
  O \left (m\alpha \log(m)^2\left ( \frac{\sqrt{\log(m\delta)}}{\eps}\right) + m \log(m\delta)  \right).
\end{equation}
matrix-vector multiplication queries to $\mA_1,\mA_2,\cdots, \mA_m$.
\end{theorem}

The proof follows by grouping all queries into a group of size $m$, implying that there is a $\log(m)$ overhead by using the tree technique. Therefore, the main alteration to Algorithm \ref{alg:main} is to 1) avoid partitioning into $1/\alpha$ subgroups and 2) calling $\texttt{Hutch++}$ at each level with updated parameters: $\eps'(\ell)  = \eps / (2^{\ell+1}\alpha \log_2 m)$ and $\delta' =  \delta/m$.

\begin{proof}[Proof of \cref{thm:alg-p-relaxed}]
 
Compared with \cref{thm:alg}, the error and success rate analysis remains unchanged. We only need to count the query complexity differently using \cref{thm:hpp-general}. Note that in this case, there is only one group of size $s = m$. As before, at each level $\ell$, we make $O(s/2^\ell )$ calls to \texttt{Hutch++}$(\mA, \eps'(\ell),\delta')$.
This leads to a total number of 
\begin{equation}
    O\left( \left(s/2^{\ell}\right) \cdot\left ( \left( \frac{\sqrt{\log(1/\delta')}}{\epsilon'(\ell)}\right)^p + \log(1/\delta')
    \right)\right)
\end{equation}
matrix-vector multiplication queries by \cref{thm:hpp-general}.  Substituting $\eps'(\ell)  = \eps / 2^{\ell+1}\alpha \log_2 m$ and $\delta' =  \delta/m$, we have
\begin{equation*}
    O \left ( m\log(m)\alpha \left ( \frac{\sqrt{\log(m\delta)}}{\eps}\right) + \frac{m}{2^\ell} \log(m\delta)  \right).
\end{equation*}

Summing over $\ell \in \{0,1, \cdots,\log_2 m\}$, we obtain that for this large group, the number of queries is bounded  by
\begin{equation*}
    O \left (  m \log^{2}(m)\left ( \frac{\sqrt{\log(m\delta)}}{\eps}\right) + m \log(m\delta)  \right).
\end{equation*}

This completes the proof.
 \end{proof}

\section{Proof Details of \texorpdfstring{\cref{sec:lower}}{Section 4}} \label{sec:sec4-apx}
\subsection{Proof of \texorpdfstring{\cref{thm:main-lowerbound}}{Theorem 4.1}}\label{sec:adap-lower-pf}
We prove the two lower bounds separately. Together they imply \cref{thm:main-lowerbound}.
\subsubsection{Lower Bound  I} 
Let $\mA \in \mathbb R^{n\times n}$ be a general square matrix.  
Recall that the goal is to estimate its trace $\Tr \mA$ up to an additive $\varepsilon \|\mA\|_p$. 
We work under the bit complexity model, where the query vectors $q_1, q_2, \cdots q_r \in \mathbb R^n$ have entries specified by $k$ bits.  To lower bound $r$, the number of queries, we reduce the communication problem of the \textsc{Approximate-Orthogonality} to trace estimation. 
 
 The \textsc{Approximate-Orthogonality} problem is a two-party communication problem defined on inputs in $\{-1,1\}^{m} \times  \{-1,1\}^{m}$  by the Boolean function 
 \begin{equation}\label{eqn:ort-def}
\operatorname{ORT}_{b, m}(x, y)= \begin{cases}1, & \text { if  }|\langle x, y\rangle| \leq b \sqrt{m} \\ -1, & \text { otherwise. }\end{cases}
\end{equation}
 The problem  is known to have $\Omega{(m)}$ communication complexity, under the uniform distribution.  Let 
 \begin{equation}
\operatorname{tail}(x)=\frac{1}{\sqrt{2 \pi}} \int_{x}^{\infty} e^{-x^{2} / 2} d x.
\end{equation}
be the tail probability of the standard normal.
%\fred{I'll modify the rest of the proof, once we can confirm if we can take $b$ to be   superconstant in the following lemma.}
 \begin{lemma}[Communication complexity of ORT, Theorem 4.2 of \cite{chakrabarti2012information}]\label{lem:CC-ORT}
 Let $b>1/5$ be a     constant and $\theta=\operatorname{tail}(2.01 \max \{66, b\})$. Then we have $\text{CC}_\theta (\operatorname{ORT}_{b, m}) = \Omega(m)$. The lower bound holds even when the inputs are drawn uniformly from $\{-1,1\}^{m} \times  \{-1,1\}^{m}$.
 \end{lemma}
 
 We now prove our adaptive trace estimation lower bound for general matrices, by connecting it with the \textsc{Approximate-Orthogonality} problem. It implies that the classic Hutchinson's estimator is optimal for constant success probability.  %The proof of the theorem can be found in \cref{sec:adap-lower-pf}.
\begin{theorem}[Adaptive query lower bound, I]\label{thm:adap-lower}
Any algorithm that accesses a square matrix $\mA$ via matrix-vector multiplication queries  requires at least $\Omega\left(\frac{1}{\epsilon^p(k+\log(1/\epsilon))} \right)$   queries  to output  an estimate $t$ such that with probability at least $1-\delta/2$, $|t - \Tr\mA| \leq \epsilon \|\mA\|_p$, for $p\in[1,2]$ and $\delta = \tail(2.01\cdot 66) = \Theta(1)$,  where the query vectors may be adaptively chosen and their  entries  are specified by $k$ bits.
 \end{theorem}
 \begin{proof}[Proof of \cref{thm:adap-lower}]
 Let $\mathcal{A}$ be a possibly adaptive algorithm for trace estimation using matrix-vector multiplication queries. Suppose it takes at most $r(n)$ queries to solve the problem, on any $n$-by-$n$ square matrix, with success rate at least $1-\delta = \Omega(1)$. 
Consider an instance of \textsc{Approximate-Orthogonality} with $b=2$, where $(x,y)$ is drawn uniformly from  $\{-1,1\}^{m} \times  \{-1,1\}^{m}$.  

The proof proceeds by reducing the problem of computing $\text{ORT}_{b,m}(x,y)$ to trace estimation via $\mathcal{A}$. 
Let $n= \frac{\delta^{p/2}}{2^{p/2} \varepsilon^p} =\Theta(1/\varepsilon^p) $ and $m = n^2$.
The reduction and its resulting communication protocol are given as follows. First, given $x\in\{-1,1\}^{m} $,  Alice creates a square matrix $\mA$, where the rows of $\mA$ correspond to the entries of $x$ in order. Similarly,  given $y\in\{-1,1\}^{m}$, Bob creates a square matrix $\mB$, where the columns of $\mB$ correspond to the entries of $y$ in order.  Then  the protocol repeats the following steps for $r(n)$ rounds. 
\begin{enumerate}[(i)] 
\item In the $i$-th round from $i=1$,  Alice creates the first query $q_i$, according to $\mathcal{A}$, given all previous query values $\left\{q_j^\top \mA\mB\right\}_{j<i}$. She computes $q_i^\top \mA$  and sends it to Bob. 
\item  Bob computes $q_i^\top \mA \mB$ and sends it back to Alice. 
\end{enumerate}
At the end of the protocol, with probability at least $1-\delta/2$, Alice and Bob obtain an estimate $t$ such that  
\begin{equation}\label{eqn:A-guarantee}
    |t - \Tr(\mA\mB)| \leq \epsilon \|\mA\mB\|_p,
\end{equation}
by the guarantee of algorithm $\mathcal{A}$. Finally, they output $z=1$ if $t \leq 3\sqrt{m} $ and $z=-1$ otherwise. 

We   argue that the above protocol   computes $\text{ORT}_{b,m}$ with error at most $\delta = \text{tail}(2.01\cdot 66)$.   First, note that by construction of steps (i) and (ii),  we have $\Tr (\mA\mB) = \langle x, y\rangle$. 
Therefore, by \cref{eqn:A-guarantee}, 
\begin{equation}\label{eqn:gap1-guarantee}
 \Pr_\mathcal{A} ( |t - \langle x,y\rangle | \leq \epsilon \|\mA\mB\|_p ) = \Pr_\mathcal{A} ( |t - \Tr (\mA\mB) |\leq \epsilon \|\mA\mB\|_p )   \geq 1-\delta/2.
\end{equation} 
It now suffices to show that the error term $\eps \|\mA\mB\|_p$ is small. Note that since $x,y$ are drawn uniformly at random, it follows that $ \E(\mA\mB)_{i,j}^2  = n$ for all $i,j \in [n]$. By linearity of expectation, $\E \|\mA\mB\|_F^2 = n^3$. By Markov's inequality, $$\Pr (\|\mA\mB\|_F^2  > t n^3) \leq 1/t$$ for any $t> 0$, and therefore,   \[\Pr (\eps \|\mA\mB\|_F  >\eps \sqrt{t} n^{3/2}) \leq 1/t.\]
Since $\|\mX\|^{p} \leq n^{1/p-1/q}\|\mX\|_q$ for any $n\times n$ matrix $\mX$, it follows that 
\[\Pr (\eps \|\mA\mB\|_p  >\eps \sqrt{t} n^{1/p -1/2}\cdot n^{3/2}) = \Pr (\eps \|\mA\mB\|_p  >\eps \sqrt{t} n^{1/p +1} ) \leq 1/t.\]
Plugging in the value of $\eps = \frac{1}{ n^{1/p}} \sqrt{\frac{\delta}{2}}$ and setting $t= 2/\delta$, we get  
\begin{equation}\label{eqn:gap2-guarantee}
\Pr_{x,y} \, (\eps \|\mA\mB\|_p  > n) \leq 1/t =  \delta / 2
\end{equation}
Combining \cref{eqn:gap1-guarantee} and \cref{eqn:gap2-guarantee} and using a union bound,
\begin{equation}
    \Pr ( | t- \langle x,y\rangle| \leq  \sqrt{m} ) \geq 1-\delta.
\end{equation}
Therefore, whenever    $\langle x, y \rangle \leq 2\sqrt{m}$, we have $t \leq 3\sqrt{m}$, and so the protocol outputs $z=1$ correctly. This proves that the protocol solves $\text{ORT}_{b,m}$ with error at most $\delta$ (for $b=2$).

To complete the proof, we account for the total communication cost of the protocol. For that, we simply note that each message from Alice or Bob is a vector of $n$ dimensions. It suffices to specify each entry with $k + \log (n/\eps)$ bits. Hence, the protocol solves $\text{ORT}_{b,m}$ with communication cost   $r(n) \cdot O(n (k+\log (1/\eps)))$. By the communication lower bound \cref{lem:CC-ORT}, it is required that
\begin{equation*}
    r(n) \cdot O(n (k+\log (1/\eps))) \geq m =  n^2.
\end{equation*}
Rearranging and using $n = \Theta(1/\eps^p)$, we have $$r(n) \geq \Omega\left(\frac{1}{\epsilon^p(k+\log(1/\epsilon))} \right),$$ as desired. This completes the proof.
 \end{proof}

 \subsubsection{Lower Bound II}
 We now give a second lower bound that yields the correct dependence on the failure probability $\delta$. The bound holds for any Schatten-$p$ norm error guarantee, so we state it generally.  In particular, we show:
 \begin{theorem}[Adaptive query lower bound, II]\label{thm:adap-lower-2}
Any     algorithm  that accesses a   square matrix $\mA$ via matrix-vector multiplication queries  requires  at least $\Omega\left(\frac{\log (1/\delta)}{k+\log\log(1/\delta)} \right)$   queries  to output  an estimate $t$ such that with probability at least $1-\delta$, $|t - \Tr\mA| \leq 0.1 \|\mA\|_p$, for any $p$ and any $\delta \in (0,1)$,  where the query vectors may be adaptively chosen and their  entries  are specified by $k$ bits.
 \end{theorem}
 
 Our proof leverages another   communication problem, \textsc{Gap-Equality}. In this problem, Alice holds $x\in \{0,1\}^n$ and Bob holds $y\in \{0,1\}^n$, under the promise that either $x=y$ or $\|x-y\|_2^2 = n/2$. They wish to compute 
 \begin{equation}
     \text{EQ}_n(x, y)= \begin{cases}1, & \text { if  } x=y \\ -1, & \text { otherwise. }\end{cases}
 \end{equation}
The problem requires linear communication complexity for any deterministic protocol \cite{buhrman1998quantum}.
 \begin{lemma}[Communication complexity of  \textsc{Gap-Equality}\label{lem:eq} \cite{buhrman1998quantum}]
 Any deterministic protocol for computing $\text{EQ}_n$ requires $\Omega(n) $ bits of communication.
 \end{lemma}
 We are now ready to prove \cref{thm:adap-lower-2}.
 \begin{proof}[Proof of \cref{thm:adap-lower-2}]
 We give a reduction from solving \textsc{Gap-Equality} as a two-party communication problem to trace estimation via adaptive matrix-vector multiplication queries. Let $n= \log(1/\delta)$ and $x,y\in \{0,1\}^n$ be an instance of \textsc{Gap-Equality}.  Let $\mA=  (x-y) (x-y)^\top$, which has rank $1$. Under the promise, either (i) $\mA= \bm{0}$, the all $0$ matrix, or (ii) has Schatten-$p$ norm $n/2$ for any $p$. In  case (ii), we have $\Tr\mA = n/2$. Thus, one can compute $\text{EQ}_n(x,y)$, by estimating $\Tr \mA$ up to an additive error of $0.1\|\mA\|_p$, for any $p$. 
 
 We now argue any trace estimation algorithm $\mathcal{A}$ with failure rate $\delta$ and error $\eps$ yields a deterministic protocol for solving $\text{EQ}_n$. 
 First, by a union bound over all possible $x,y$ under the promise, we have that for all  $\mA=  (x-y) (x-y)^\top$, the output $t$ of $\mathcal{A}$ given $\mA$ always satisfies 
 \begin{equation}\label{eqn:thm44}
     | t- \Tr \mA|  \leq \eps \|\mA\|_p.
 \end{equation}
 Suppose $\mathcal{A}$  uses  $r = o(\log(1/\delta))$ adaptive queries $q_1, q_2, \cdots, q_r$. In case (i) when $\mA=\bm{0}$, all query answers it receives are the zero vector. The algorithm must always output $0$, to satisfy the   trace estimation guarantee (\cref{eqn:thm44}). 
 Thus, in order to always be correct in case (ii), it must be that one of its query answers is not $0$. 
 But as soon as its first query answer is not $0$, it knows that it is in case (ii). It follows that algorithm $\mathcal{A}$ just keeps receiving the all-$0$ vector until it either decides to stop querying or receives a non-zero output vector and immediately decides to stop querying. Thus, for these inputs, we can assume the query algorithm is in fact non-adaptive, since we can consider what its query sequence would be in advance if it were to repeatedly receive the $0$ vector as an answer. Hence, we can think of $\mQ = (q_1, q_2,\cdots, q_r)$ as an $r \times n$ matrix with entries specified with $k$ bits, and we have the property that $\mQ(x-y) = 0$ if and only if $x = y$.  This gives a protocol for \textsc{Gap-Equality}: Alice simply sends $\mQ x$ to Bob, who checks if $\mQ x = \mQ y$. The communication is $$r(k + \log n) = r(k + \log \log(1/\delta)),$$ which must be $\Omega(\log(1/\delta))$ by \cref{lem:eq}, and so we get an $$r = \Omega(\log(1/\delta)/(k + \log \log(1/\delta)))$$ adaptive lower bound.

 \end{proof}
 
 \subsection{Proof of \texorpdfstring{\cref{thm:lb_small}}{Theorem 4.2}}\label{sec:pf-lb-small}

We start with a standard definition. 
\begin{definition}[Gaussian and Wigner Random Matrices]
\label{def:gaussian_wigner}
    We let $\mG \sim \gN(n)$ denote an $n \times n$ random Gaussian matrix with i.i.d. $\gN(0, 1)$ entries. We let $\mW \sim \gW(n) = (\mG + \mG^T)/2$ denote an $n \times n$ Wigner matrix, where $\mG \sim \gN(n)$. 
\end{definition}

\begin{fact}[Upper and Lower Gaussian Tail Bounds]
\label{fact:gaussian_tail}
Letting $Z \sim \gN(0, 1)$ be a univariate Gaussian random variable, for any $t > 0$, \[\Pr[|Z| \geq t] = \Theta(t^{-1}\exp(-\frac{t^2}{2})).\]
\end{fact}

Suppose that we draw a matrix $\mG \in \R^{n \times n}$ from the Gaussian or related Wigner distribution and try to learn the entries of the matrix via matrix-vector queries. Because the Gaussian is rotationally and subspace invariant, after a few queries, the conditional distribution of the remaining matrix is also Gaussian (or Wigner)-distributed, no matter how the queries are chosen. This property allows us to exactly characterize the remaining uncertainty of the trace estimation procedure, especially with respect to the failure probability $\delta$, even after seeing a few query results. 

\begin{lemma}(Conditional Distribution [Lemma 3.4 of \cite{simchowitz2018tight}])
\label{lem:conditional}
Let $\mG \sim \mathcal{N}(n)$ be as in \cref{def:gaussian_wigner} and suppose our matrix is $\mW = (\mG + \mG^\top)/2$. Suppose we have any sequence of vector queries, $\vv_1,..., \vv_T$, along with responses $\vw_i = \mW \vv_i$. Then, conditioned on our observations, there exists a rotation matrix $\mV$, independent of $\vw_i$, such that 

$$ \mV\mW\mV^\top = \begin{bmatrix}
Y_1 & Y_2^\top \\
Y_2 & \widetilde{\mW} \end{bmatrix},$$
where $Y_1, Y_2$ are deterministic and $\widetilde{\mW} = (\widetilde{\mG} + \widetilde{\mG}^\top)/2$, where $\widetilde{\mG}  \sim \mathcal{N}(n- T)$.
\end{lemma}

We are now ready to prove \cref{thm:lb_small}.
\begin{proof}[Proof of \cref{thm:lb_small}]
By standard minimax arguments, it suffices to construct a hard distribution for any deterministic algorithm. Consider $\mW \sim \mathcal{W}(n)$ for some $n$ that we will determine later. From concentration of the singular values of large Gaussian matrices \cite{rudelson2010non},  with probability at least $1 - \delta/10$, we have $ \sigma_{\max}(\mG) \leq Cn^{1/2}$ for some absolute constant $C$ when $n \geq \log(1/\delta)$. Therefore, we conclude that $\|\mG\|_{p} \leq Cn^{1/2 + 1/p}$ for some absolute constant $C$. Therefore, by the triangle inequality, $\|\mW\|_p$ can be bounded by the same value.

Let $m$ be the number of matrix-vector queries, and assume that $m \leq n/2$. By \cref{lem:conditional}, we see that conditioned on the queries, our matrix $\mW$ can be decomposed into a determined part and a Gaussian submatrix $\widetilde{\mW} \sim \mathcal{W}(n- m)$. Therefore, our conditional distribution of the trace of $\mW$ is, up to a deterministic shift, the same as the distribution of $\widetilde{\mW}$, which is simply a Gaussian with variance at least  $n - m \geq n/2$. We can check this since \[\tr\left(\widetilde{\mW}\right) =\frac{1}{2} \tr\left(\widetilde{\mG}\right) + \frac{1}{2}\tr\left(\widetilde{\mG}^\top\right) = \tr \left(\widetilde{\mG}\right) = \sum_{i} \widetilde{\mG}_{ii},\] where $\mG_{ii} \sim N(0, 1)$ are independent for $1 \leq i \leq n -m $.

Since our algorithm determines a Gaussian of variance at least $n - m \geq n/2$ up to an additive error of $\epsilon \|\mA\|_p$ with probability at least $1-\delta$, we conclude that if $\epsilon \|\mA\|_p \leq \sqrt{\log(1/\delta)n}$, then we have a contradiction from the anti-concentration of Gaussians (see \cref{fact:gaussian_tail}). Therefore, whenever $\epsilon \|\mA\|_p \leq \sqrt{\log(1/\delta)n}$ holds, we can deduce a lower bound on the number of matrix-vector queries: $m \geq n/2$. 

Therefore, solving $\epsilon \|\mA\|_p \leq \sqrt{\log(1/\delta)n}$ for the largest possible value of $n$ gives: 

$$n = \Omega\left( \left(\frac{\sqrt{\log(1/\delta)}}{\epsilon}\right)^p\right) $$ 

Note that this holds for any $\delta, \epsilon > 0$ such that $n \geq \log(1/\delta)$. Therefore, we need to enforce that $\epsilon < (\log(1/\delta))^{1/2 - 1/p}$. 
\end{proof}

\section{Proof Details for \texorpdfstring{\cref{sec:lower-dyn}}{Section 5}}
\subsection{Proof of \texorpdfstring{\cref{lem:cond-lb}}{Theorem 5.1}}\label{sec:cond-pf}
\begin{proof}[Proof of \cref{lem:cond-lb}]
Let $\alpha = 1/(m-1)$. Given a square matrix $\mA$ with $\|\mA\|_p=1$, construct a sequence of matrices 
\begin{equation}
\mA_{1}=0, \quad \mA_{2}=\alpha \cdot \mA, \quad \ldots \quad \mA_{1 / \alpha}=(1-\alpha) \mA, \quad \mA_{m}=\mA.
\end{equation}
Suppose that we have a dynamic trace estimation algorithm $\mathcal A$ running on the sequence   $(\mA_i)$. By construction, each $\mA_i$ is a scaling of $\mA$.  Suppose that  in the end $\mathcal{A}$ outputs an estimate  $t_m$ such that $|t_m - \Tr \mA| \leq \eps \|\mA\|_p$ with probability at least $1-\delta$, using matrix-vector multiplies with $\mA$. This solves the static trace estimation problem with a Schatten-$p$ norm error guarantee. By assumption, it must have used  $\Omega(r)$ matrix-vector multiplication queries with respect to $\mA$. 
Therefore, if $\mathcal{A}$ uses $o(r\alpha m)$ queries, it would immediately violate our assumption, which is  a contradiction.
\end{proof}

\subsection{Proof of \texorpdfstring{\cref{thm:dynamic-hard-instance}}{Theorem 5.4}}\label{sec:pf-dynamic-hard-instance} 
\begin{proof}[Proof of \cref{thm:dynamic-hard-instance}]
Let $x,y \in \{0,1\}^n$ be an instance of \textsc{Gap-Equality}, where $n= \log(1/\delta)$. Recall that \textsc{Gap-Equality} is a promise problem. Under its promise, either $x=y$ or $\|x-y\|^2_2 = n/2$, and the goal is to distinguish the two cases. For any given $x,y$, let $\mB_{x,y} = \frac{2}{n} (x-y)(x-y)^\top$. Then since $\mB_{x,y}  $ is rank-$1$, $\|\mB_{x,y}\|_p = 0$ if $x=y$ or $\|\mB_{x,y}\|_p=1$ otherwise.   

To obtain the claimed lower bound, we consider two parameter regimes. First, if $\alpha > \eps$, we construct the following hard instance, which is a sequence of $m$ matrices satisfying the Schatten $p$ norm assumption for dynamic trace estimation. Let $\mA_0 \in \mathbb{R}^{N\times N}$ be an all $0$s matrix, with $N= \min \{m, 1/\alpha\} \log(1/\delta)$.  Throughout the updates, $\mA_i$ will remain a block diagonal matrix, which consists of $m$ block matrices along the diagonal and each of dimension $\log (1/\delta) \times \log(1/\delta)$. In particular, for all steps $i = \{1,2,\cdots, \min \{m, 1/\alpha\}-1\}$, we set 
\begin{equation}\label{eqn:insert-dy}
\mA_i = 
\begin{bmatrix} 
  \mB_1 & 0         &\cdots &\cdots    & \cdots& \cdots & 0            \\
  0            & \mathbf{B}_2  &   0 &\cdots & \cdots& \cdots & 0            \\
  \vdots       & 0        & \ddots    &\cdots &\cdots & \cdots & \vdots       \\
  \vdots       & \vdots        & 0    &\mB_i &0 & \cdots & \vdots       \\
    \vdots       & \vdots        & \vdots    &0 &0  & \cdots& \vdots       \\
        \vdots       & \vdots        & \vdots   &\vdots &\vdots& \ddots & \vdots       \\ 
  0            & 0         &\cdots &\cdots      & \cdots& \cdots & 0
\end{bmatrix}
\end{equation}
where $\mB_i =\alpha \mB_{x_i, y_i}$ with $x_i,y_i \in \{0,1\}^n$ an independent instance of \textsc{Gap-Equality}. In other words, at each step $i$, we update $\mA_{i-1}$ by replacing the $i$-th diagonal block (currently being all $0$s) with $\mB_i$.  
Each update changes the trace by $0$ or $\alpha$, by the construction of $\mB_{x,y}$. If $m \leq 1/\alpha$, this completes the construction, and note that the matrices $\{\mA_i\}$ all have norm bounded by $1$. If   $m > 1/\alpha$, we continue the construction by deleting one distinct diagonal block at each step until the matrix is the zero matrix. Then we repeat the same rounds of insertion (according to \cref{eqn:insert-dy}) and deletion until reaching time step $m$. Observe again that the construction satisfies the Schatten norm assumption for dynamic trace estimation.

We now argue the query complexity as follows:
\begin{itemize}
    \item In the case of $m \leq 1/\alpha$, each update is either (i) trivially $0$ or  (ii) increases the trace by $\alpha > \eps$. Hence, any dynamic algorithm for outputting $|t_i - \Tr \mA_i| \leq \eps < \alpha$, with   probability at least $1-\delta$, would distinguish between case (i) and (ii) with probability at least $1-\delta$. However, this requires $$\Omega\left(\frac{\log (1/\delta)}{k+\log\log(1/\delta)} \right)$$ matrix-vector multiplication queries by \cref{thm:adap-lower-2}. 
    \item In the case of $m \leq 1/\alpha$, note that (almost) half of the update steps are insertions. By the same argument, any dynamic algorithm that gives a good estimate in an insertion step $i$ can solve the hard instance of estimating $\Tr \mB_i$. Hence, we get the same query complexity lower bound.
\end{itemize}
To summarize, if $\alpha > \eps$, we get a lower bound of  $\Omega\left(m \cdot \frac{\log(1/\delta)}{k+ \log\log(1/\delta)} \right)$ queries. 

Now we move on to the case of $\alpha \le \eps$. We use the same construction as described by \autoref{eqn:insert-dy}, where each $\mA_i$ consists of multiple updates over $s=\lceil\eps/\alpha\rceil$ steps by setting $\mB_i = \sum_{j=1}^{s} (1/s)\cdot \mB_{x_{i},y_{i}}$ with $\mB_{x_{i},y_{i}}$ an independent instance of \textsc{Gap-Equality}.  We repeat the argument earlier and apply the hardness of \cref{thm:adap-lower-2} on the sequence of $\mA_i$. This blows up the sequence length by a factor of $s$, and hence leads to a lower bound of  $$\Omega\left(m \cdot \frac{\alpha}{\eps} \frac{\log(1/\delta)}{k+ \log\log(1/\delta)} \right).$$
This finishes the proof.
\end{proof} 
 \section{Lower Bound for Non-Adaptive Trace Estimation}\label{sec:nonad}
%  \fred{This section probably has to go to appendix for the NeurIPS submission.}
In the case of non-adaptive queries, we give a stronger lower bound than \cref{thm:main-lowerbound} in the bit complexity model. The bound matches Hutchinson's guarantee for general square matrices up to a bit complexity term.
\begin{theorem}[Non-adaptive query lower bound]\label{thm:non-adaptive}
Any     algorithm  that accesses a   square matrix $\mA$ via non-adaptive matrix-vector multiplication queries  requires  at least $\Omega\left(\frac{\log^{p/2} (1/\delta)}{\eps^p(k+\log(1/\eps))} \right)$   queries  to output  an estimate $t$ such that with probability at least $1-\delta$, $|t - \Tr\mA| \leq \eps \|\mA\|_p$, for any $p$ and $\eps, \delta  \in (0,1) $,  where each entry of the query vectors  is specified by $k$ bits.
\end{theorem}
The proof is via a reduction from the Augmented Indexing   communication problem with low error \cite{jayram2013optimal}. For a sufficiently large universe $\mathcal U$ and an element $\perp\notin \mathcal{U}$, the problem $\textsc{IND}_{n,\mathcal{U}}$ is defined as follows.
\begin{itemize}
    \item Alice gets $x =\left(x_{1}, x_{2}, \ldots, x_{n}\right) \in \mathcal{U}^{n}$.
    \item Bob gets $y=\left(y_{1}, y_{2}, \ldots, y_{n}\right) \in(\mathcal{U} \cup\{\perp\})^{n}$ such that for some unique $i$
    \begin{enumerate}[(i)]
        \item $y_{i} \in \mathcal{U}$,
        \item $y_{k}=x_{k}$ for all $k<i$,
        \item $y_{i+1}=y_{i+2}=\cdots=y_{N}=\perp$.
    \end{enumerate}
\end{itemize}
Finally, Bob wishes to output whether $x_i=y_i$. The one-way communication complexity of $\textsc{IND}_{n,\mathcal{U}}$ is known:
\begin{lemma}[Communication complexity of Augmented Indexing \cite{jayram2013optimal}]\label{lem:index-cc}
Any one-way communication protocol for computing  $\textsc{IND}_{n,\mathcal{U}}$ with error $\delta \le \frac{1}{4|\mathcal{U}|}$ requires at least $n\log|\mathcal{U}|/2$ bits of communication.
\end{lemma}

We now describe how to solve $\textsc{IND}_{n,\mathcal{U}}$ in one round of  communication via a non-adaptive trace estimation protocol. 
\begin{proof}[Proof of \cref{thm:non-adaptive}]
Let $\kappa = 1/4\delta^{p/2}$, $n=\left(\sqrt{\log (3/\delta)}/\eps\right)^p$, $m= c/\left(4\delta^{p/2} \eps^p\right)$  for $c>0$ a small enough constant, and $\mathcal{U} =[\kappa ]$. In the following, we view $\mathcal{U}$ equivalently as the collection of one-hot encodings, i.e., $1$-sparse vectors in $\{0,1\}^\kappa$. Let $x,y$ be an instance of  $\textsc{IND}_{n,\mathcal{U}}$ and $i$ be the special index under the promise.  Given Alice's input $x \in \{0,1\}^{ 1/\eps^2 \times \kappa}$ and $\eps, \delta \in (0,1/4)$, we construct an  $n\times n$ real square matrix $\mA$, as follows.
\begin{itemize}
    \item  Let $\mB \in \{0,1\}^{m\times m}$ have all rows but the $i$-th row being the all-zeros vector;
    \item The $i$-th row of $\mB$ is the vector $v= x$ (with precisely $c/\eps^p$ non-zero entries).
    \item Let $\mA = \frac{1}{n}\mG \mB\mG^\top$, where $\mG \in \mathbb{R}^{n\times m}$ is a random matrix with i.i.d.\ standard Gaussian entries.
\end{itemize}
To solve $\textsc{IND}_{n,\mathcal{U}}$, it suffices for Bob to recover $v_i$ with probability at least $1-\delta$. By construction, we immediately have that $\Tr \mB = v_i$ and $\|\mB \|_F= \sqrt{c}/\eps^{p/2}$.  Moreover, by the guarantee of Hutchinson's estimator (see, e.g., Lemma 2 of \cite{mmmw2020hutch_pp}), \[|\Tr \mA - \Tr \mB | \leq {\eps^{p/2}}{\left(\log(1/\delta)\right)^{1/2 -p/2}} \|\mB\|_F=\sqrt{c} (\log(1/\delta))^{1/2 -p/2} \leq \sqrt{c}\] 
with probability at least $1-\delta/3$. By the Johnson-Lindenstrauss lemma, $ \|\mA\|_F \leq   \sqrt{c}$ with probability $1-\delta/3$. By construction, $\mB$ has rank one and so $\mA$ has rank one. It follows that $\|\mA\|_p =  \|\mA\|_F \leq   \sqrt{c}$ for any $p$.

Now suppose that there is   a non-adaptive trace estimation protocol that has $\eps$ approximation error and $\delta/ 3$ failure rate, using $r$ queries $\mQ =(q_1,q_2,\cdots, q_r)$.   
To finish the reduction, Alice sends matrix $\mA \mQ$ to Bob. Bob can obtain an estimate $t$ such that  with probability at least $1-\delta /3$, $|t - \Tr \mA| \leq \eps \|\mA\|_p$.
Now taking a union bound and applying the triangle inequality, we have that with probability $1-\delta$, 
\begin{align*}
   |t-v_i| = |t -\Tr \mB  | & \leq |t- \Tr \mA| + |\Tr \mA - \Tr \mB| \\ 
   &\leq   \sqrt{c} + \sqrt{c}\eps\\ &<1/2,
 \end{align*}
for $\eps <1/4$ and a sufficiently small $c$ (say, $c<0.01$). Hence, Bob can recover $v_i$ and compute $\textsc{IND}_{n,\mathcal{U}}$.

On the other hand, by \cref{lem:index-cc}, there is a communication lower bound of $\Omega((\log(1/\delta)/\eps^2)$ bits for the problem.  Each entry of $v\mQ$ is specified by $ O(\log(1/\eps) + k) $ bits, so the total communication of sending $\mA\mQ$ is $O(\log(1/\eps) + k) \cdot r$. This leads to a query lower bound of \[r \geq \Omega((\log(1/\delta))/(\eps^2 (\log(1/\eps) + k))),\] as claimed.
\end{proof}

\section{Experimental Details}\label{sec:exp_details} 
\subsection{Experimental Results on Synthetic Data}\label{sec:exp-synth}
We follow a similar experimental set-up as in \cite{dharangutte2021dynamic} and consider  small and large perturbations. 
We also report the average absolute error over all time steps and all trials.
In the small perturbation regime, our algorithm achieves errors   (average error: $0.0104$)  that are   negligible in comparison with DeltaShift (average error: $1.9804$) and other procedures.
In the high perturbation regime, our algorithm (average error: $1.6607$)  outperforms   Hutchinson's and Diffsum and is comparable with Deltashift (average error: $1.5868$).
We notice, across a variety of regimes,  that Hutchinson's estimator and Diffsum tend to accumulate estimation error over the dynamic updates, whereas our algorithm  and DeltaShift remain stable.
\begin{figure}[h]
\centering
\begin{subfigure}{.5\textwidth}
  \centering
  \includegraphics[scale=0.4]{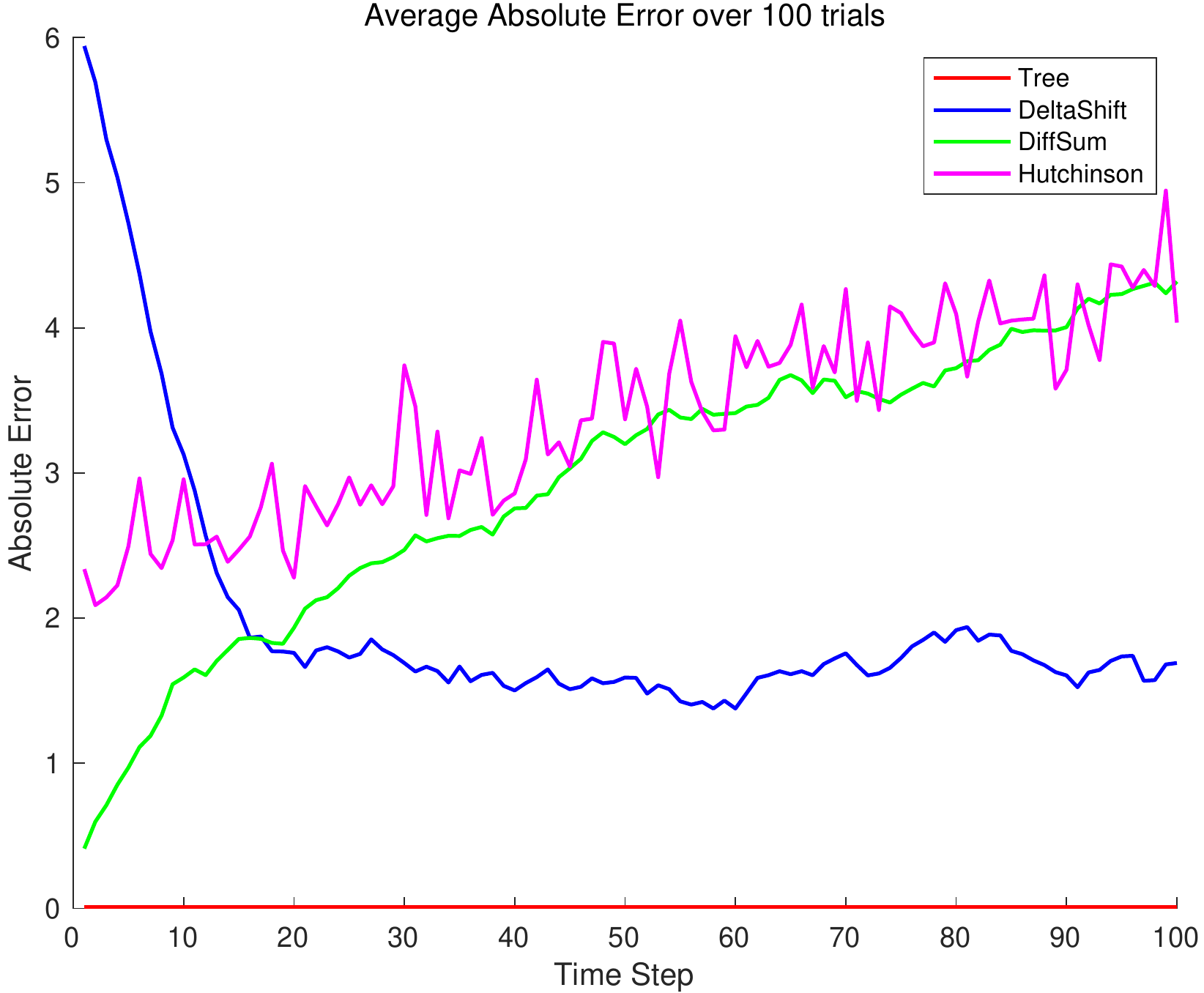}
  \caption{Low perturbation}
  \label{fig:sub1}
\end{subfigure}%
\begin{subfigure}{.5\textwidth}
  \centering
  \includegraphics[scale=0.4]{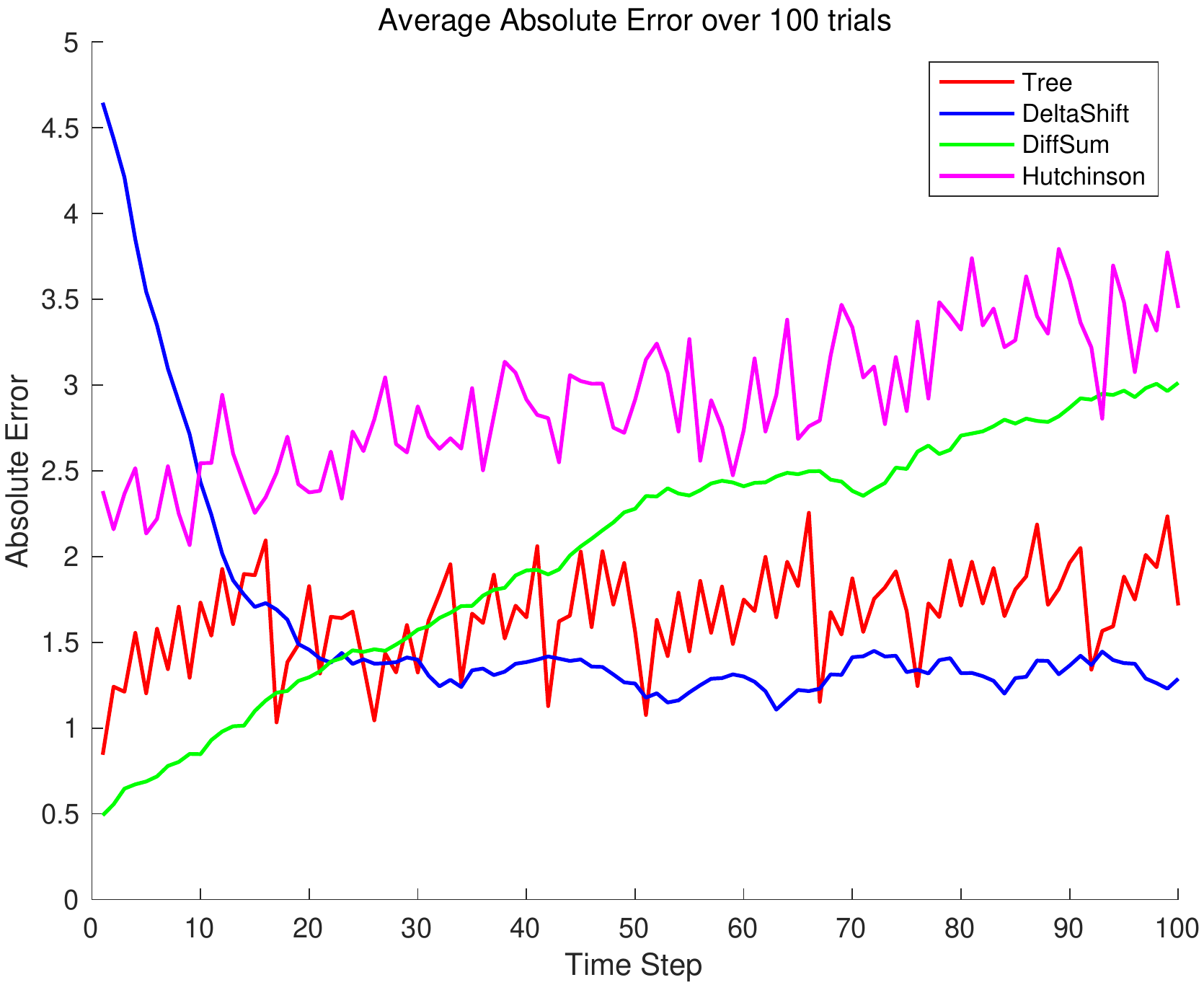}
  \caption{High perturbation}
  \label{fig:sub2}
\end{subfigure}
\caption{Synthetic data, Tree refers to our algorithm (Algorithm \ref{alg:main}). Query budget is $8,000$. We measure the error at step $i$ simply by absolute error $|t_i - \Tr \mA_i|$ }
\label{fig:test}
\end{figure}

\subsection{Experimental Setup}
\paragraph{Allocation of query budget.}
We allocate the same query budget in each time step of DeltaShift and in Hutchinson's estimator. For DiffSum, we allocate $1/5$ of the budget for estimating $\Tr\mA_1$ and an equal number of queries among the remaining steps. To optimize performance,   the number of groups in our algorithm is tuned.  

\paragraph{Experiments on synthetic data.} On both small and large perturbation experiments, we choose  the dimension to be $n=1000$. The first matrix $\mA_1$ in the sequence is a  symmetric matrix with random (unit-norm) eigenvectors and eigenvalues drawn uniformly from $[-1,1]$.  In the small perturbation regime, a random rank-$1$ matrix $\Delta_j = 5e^{-5}rgg^\top$ is  added in each time step, where $r$ is a random sign and $g$ is a   standard Gaussian in $n$ dimensions. 
In the large perturbation regime, each update is a random rank-$20$ positive semidefinite matrix.

\paragraph{Neural network weight matrices.} 
The network consists of two hidden layers of the same size, with standard ReLU activations. The mini-batch size is set to $60$ and learning rate is set to $0.01$.

We optimized the performance of our trace estimation algorithm by choosing its number of groups to be $20$.
\end{document}